\documentclass[onecolumn]{article}%

\usepackage[caption=false]{subfig}

\usepackage{standalone}

\usepackage{lineno,hyperref}
\modulolinenumbers[5]

\usepackage{float}
\usepackage{amssymb,amsmath,amsxtra,amsfonts,amsthm,mathtools}
\usepackage{graphicx}

\usepackage{array,multirow,makecell}
\usepackage{subfig}
\usepackage{nopageno}

\usepackage{tikz}

\usepackage{glossaries}

\usepackage{authblk}

\newtheorem{theorem}{Theorem}[section]

\newtheorem{algorithm}[theorem]{Algorithm}

\newtheorem{corollary}[theorem]{Corollary}

\newtheorem{definition}[theorem]{Definition}
\newtheorem{example}[theorem]{Example}

\newtheorem{remark}[theorem]{Remark}

\begin{document}


\title{On The Block Decomposition and Spectral Factors of $\lambda$-Matrices}

\author{Belkacem Bekhiti\textsuperscript{1},Abdelhakim Dahimene\textsuperscript{1},Kamel Hariche\textsuperscript{1}  and  George F. Fragulis\textsuperscript{2} }
\affil[]{\textsuperscript{1}Signal and System Laboratory, Electronics and Electrotechnics Institute University of Boumerdes, Algeria, IGEE Ex:(INELEC)\\ %
\affil[]\textsuperscript{2}Western Macedonia Univercity of Applied Sciences, Kozani, Hellas. \\ %
	 
}


\maketitle

\begin{abstract}
In this paper we factorize matrix polynomials into a complete set of spectral factors using a new design algorithm 
and we provide a complete set of block roots (solvents). 
The  procedure is  an extension of the (scalar) Horner method  for the computation of the block roots of matrix polynomials.
The Block-Horner method brings an iterative nature, faster convergence, nested programmable scheme, needless of any prior knowledge 
of the matrix polynomial. 
In order to avoid the initial guess method we proposed a combination of two computational procedures . 
First we start giving  the right Block-\emph{Q. D. (Quotient Difference)} algorithm for spectral decomposition and matrix polynomial factorization.
Then the construction of new block Horner algorithm for extracting the complete set of spectral factors is given.
\end{abstract}


 
 \section{Introduction}
 In the early days of control and system theory, frequency domain techniques were the principal tools of analysis, modeling and design for linear systems. Dynamic systems that can be modeled by a scalar $m^{th}$ order linear differential (difference) equation with constant coefficients are amenable to this type of analysis see [\cite{10}] and [\cite{38}]  - other references [\cite{46}].
 In the case of a single input - single output (SISO) system  the transfer function is a ratio of two scalar polynomials. The dynamic properties of the system (time response, stability, etc.) depend on the roots of the denominator of the transfer function or in other words on the solution of the underlying
 homogeneous differential equation (difference equation in discrete-time systems) [\cite{01}],[\cite{40}],[\cite{41}] [\cite{47}],[\cite{48}],[\cite{49}]. 
 The denominator of such systems is a scalar polynomial and its spectral characteristics depend on the location of its roots in the s-plane. Hence the factorization (root finding) of scalar polynomials is an important tool of analysis and design for linear systems[\cite{09}] . 
 In the case of  multi input - multi output (MIMO) systems the dynamics can be modeled by high-degree coupled differential equations or ${l^th}$ degree $m^{th}$ order vector linear differential (difference) equations with matrix constant coefficients.

 In this paper, we treat the dynamic properties of  multivariable systems using the latent roots and /or the spectral factors of the corresponding matrix polynomial, following the research by a number of recent publications see [\cite{11}], [\cite{12}], [\cite{13}], [\cite{14}], [\cite{15}], [\cite{16}] and [\cite{27}] .\vspace{0.2cm}

 The algebraic theory of matrix polynomials has been investigated by Gohberg et al.[\cite{32}], Dennis et al. [\cite{28}], Denman [\cite{30}],[\cite{31}], Shieh et al. [\cite{17}],[\cite{19}], [\cite{21}], [\cite{36}] and Tsai et al.[\cite{20}]. Various computational algorithms [\cite{28}], [\cite{31}], [\cite{32}], [\cite{17}] and [\cite{23}] are available for evaluating the solvents and spectral factors of a matrix polynomial. Recent approaches [\cite{28}],[\cite{29}],[\cite{22}] are the use of the eigenvalues and eigenvectors of the block companion form matrix of the corresponding matrix polynomial (the$\lambda$-matrices) to construct the solvents of that matrix polynomial based in the use of solvents.

 It is often inefficient to explicitly determine the eigenvalues and eigenvectors of a matrix, which can be
 ill conditioned and either non-defective or defective. On the other hand, without prior knowledge of the eigenvalues
 and eigenvectors of the matrix, the Newton-Raphson method [\cite{17}],[\cite{18}] has been successfully utilized for finding the
 solvents, as well as  the block-power method (Tsai et al.) [\cite{20}] for finding the solvents and spectral
 factors of a general nonsingular (monic or comonic)polynomial matrix. 
 
 The matrix polynomial must have distinct block solvents, and the convergence rate of the power method depends strongly on
 the ratio of the two block eigenvalues of largest magnitude [\cite{25}]. There are numerous numerical methods
 for computing the block roots of matrix polynomials without any prior knowledge of the eigenvalues and eigenvectors
 of the matrix polynomial. The most efficient and more stable one that give the 
 complete set of solvent at time is the $Q.D$. (quotient-difference) algorithm. The use of the $Q.D$. algorithm for such purpose has been
 suggested by K. Hariche [\cite{04}].\vspace{0.2cm}
 

 The purpose of this paper is to illustrate the so called Block quotient-difference ($Q.D$.) algorithm
 and  extend the (scalar) Horner method to its block form for the computation of the block roots of matrix polynomial and the determination of complete set of solvents and spectral factors of a 
 monic polynomial, without any prior knowledge of the eigenvalues and eigenvectors of the matrix. See also  Pathan and Collyer [\cite{34}]) where there is a presentation on Horner's method and its application in solving polynomial equations by determining the location of roots. 
 
 The objectives of this paper are described as follows: 
 
 \begin{itemize}
 	\item Illustration and finalization of the Block quotient-difference ($Q.D$.) algorithm for spectral decomposition and matrix polynomial factorization.
 	\item Construction of a new block-Horner array and block-Horner algorithm for extracting the complete set of spectral factors of matrix polynomials.
 	\item Combined the above algorithms for fast convergence, high stability and avoiding initial guess.\vspace{0.2cm}\\
 \end{itemize}

 	\section{Preliminaries}
 	In this section we give some background material.
 	\subsection{Survey on matrix polynomials}
 	
 	\begin{definition} Given the set of  $m\times m$  complex matrices $A_0,A_1,...,A_\emph{l}$, the following matrix valued function of the complex variable $\lambda$  is called a matrix polynomial of degree $l$  and order $m$:
 		\begin{equation}\label{eq.1}
 		A(\lambda)= A_0\lambda^l+A_1\lambda^{l-1}+...+A_{l-1}\lambda+A_l
 		\end{equation}
 	\end{definition}
 	
 	\begin{definition} The matrix polynomial $A(\lambda)$ is called:\\
 		i. Monic if is $A_0$ the identity matrix.\\
 		ii. Comonic if $A_l$ is the identity matrix.\\
 		iii. Regular or nonsingular if $det(A(\lambda))\neq0$ .\\
 		iv. Unimodular if $det(A(\lambda))$ is nonzero constant. 
 	\end{definition}

 	\begin{definition} The complex number $\lambda_{i}$ is called a latent root of the matrix polynomial $A(\lambda)$ if it is a solution of the scalar polynomial equation $\mbox{det}(A(\lambda))=0$ The nontrivial vector  $p$ , solution of  $A(\lambda_{i})p=0_m$
 		is called a primary right latent vector associated with $\lambda_{i}$. Similarly the nontrivial vector $q$ solution of $q^TA(\lambda_{i})=0_m$ is called a primary left latent vector associated with $\lambda_{i}$.
 	\end{definition} 
 	
 	\begin{remark}. If $A(\lambda)$  has a singular leading coefficient $(A_l)$ then $A(\lambda)$ has latent roots at infinity. From the definition we can see that the latent problem of a matrix polynomial is a generalization of the concept of eigenproblem for square matrices. Indeed, we can consider the classical eigenvalues/vector problem as finding the latent root/vector of a linear matrix polynomial  $(\lambda I-A)$ .
 		We can also define the spectrum of a matrix polynomial $A(\lambda)$  as being the set of all its latent roots (notation $\sigma(\lambda)$  ).
 	\end{remark} 
 	
 	\begin{definition} 
 		A right block root also called solvent of $\lambda\mbox{-matrix} A(\lambda)$  and is an $m\times m$ real matrix $R$ such that:
 	\end{definition} 
 	
 	\begin{equation}
 	\begin{split}
 	R^l+A_{1}R^{l-1}+...+A_{l-1}R+A_{l}=O_{m} \hspace*{1cm}\\ \Leftrightarrow A_R(R)=\sum\limits_{i=0}^{l}A_{i}R^{l-i}=O_{m} \hspace*{2cm}
 	\end{split}
 	\end{equation}
 	While a left solvent is an $m\times m$ real matrix $L$ such that:
 	\begin{equation}
 	\begin{split}
 	L^l+L^{l-1}A_{1}+...+LA_{l-1}+A_{l}=O_{m} \hspace*{1cm}\\ \Leftrightarrow A_L(L)=\sum\limits_{i=0}^{l}L^{l-i}A_{i}=O_{m} \hspace*{2cm}
 	\end{split}
 	\end{equation}
 	
 	\begin{remark} From [\cite{26}] we have the following: 
 		\begin{itemize}
 			\item Solvents of a matrix polynomial do not always exist.
 			\item Generalized right (left) eigenvectors of a right (left) solvent are the generalized latent vectors of the corresponding matrix polynomial
 		\end{itemize}
 	\end{remark}
 	
 	\begin{definition}: A matrix $R$ (respectively: $L$) is called a right (respectively: left) solvent of the matrix
 		polynomial if and only if the binomial $(\lambda I-R)$(respectively:$(\lambda I-L$))divides exactly $A(\lambda)$  on the right (respectively: left).
 	\end{definition} 
 	From [\cite{05},\cite{06})],\cite{07})],\cite{08})] we have 
 	
 	\begin{theorem}\label{theorem1} Given a matrix polynomial
 		\begin{equation}
 		A(\lambda)=A_0 \lambda^l+A_1 \lambda^{l-1}+...+A_{l-1}\lambda+A_l
 		\end{equation}
 		a) The reminder of the division of $A(\lambda)$ on the right by the binomial $(\lambda I-X)$ is $A_{R}(X)$ \\
 		b) The reminder of the division of $A(\lambda)$ on the left by the binomial $(\lambda I-X)$ is $A_{L}(X)$ \\
 		Hence there exist matrix polynomials $Q(\lambda)$  and $S(\lambda)$  such that:
 		\begin{equation}
 		\begin{array}{ccc}
 		A(\lambda)&=&Q(\lambda)(\lambda I-X)+A_R (X)\\ [0.1cm]
 		&=&(\lambda I-X)S(\lambda)+A_L (X)
 		\end{array}
 		\end{equation}
 	\end{theorem}

 	\begin{corollary} The fundamental relation that exist between right solvent (respectively: left solvent) and right (respectively: left) linear factor:
 		\begin{equation}
 		\begin{array}{ccc}
 		A_R (X)&=&0 \hspace{0.1cm}\mbox{iff} \hspace{0.1cm}A(\lambda)=Q(\lambda)(\lambda I-X)\\ [0.1cm]
 		A_L (X)&=&0 \hspace{0.1cm}\mbox{iff} \hspace{0.1cm}A(\lambda)=(\lambda I-X)S(\lambda)
 		\end{array}
 		\end{equation}
 	\end{corollary}
 	
 	In reference [2] it is stated the following:
 	
 	\begin{theorem}: Consider the set of solvents \{$R_1,R_2,...,R_l$\} constructed from the eigenvalues ($\lambda_1,\lambda_2,...,\lambda_l$) of a matrix $A_c$. \{$R_1,R_2,...,R_l$\}is a complete set of solvents if and only if:
 		\begin{equation}
 		\left\{
 		\begin{array}{ll}
 		\cup \sigma(R_i )=\sigma(A_c )\\
 		\sigma(R_i )\cap \sigma(R_j )=\varnothing\\
 		det(V_R(R_1,R_2,...,R_l))\neq0
 		\end{array}
 		\right.
 		\end{equation}
 		Where:
 		$\sigma$ denotes the spectrum of the matrix.
 		$V_R$ Vandermonde matrix corresponding to \{$R_1,R_2,...,R_l$\} given as
 		\begin{equation}
 		V_R(R_1,R_2,...,R_l)=
 		\left( \begin{matrix}
 		I_m\hspace{0.5cm}I_m \hspace{0.4cm}...\hspace{0.4cm}I_m \\
 		R_1\hspace{0.5cm}R_2\hspace{0.4cm}...\hspace{0.4cm}R_l \\
 		\vdots\hspace{0.6cm}\vdots \hspace{0.7cm}...\hspace{0.4cm}\vdots\\
 		R_1^{l-1}\hspace{0.3cm}R_2^{l-1}\hspace{0.1cm}...\hspace{0.2cm}R_l^{l-1}
 		\end{matrix} \right)
 		\end{equation}
 	\end{theorem}
 	
 	\begin{remark}: we can define a set of left solvents in the same way as in the previous theorem.
 		The relationship between latent roots, latent vectors, and the solvents can be stated as follows:
 	\end{remark}
 	
 	From [\cite{24}] we have the following:
 	
 	\begin{theorem}:If $A(\lambda)$  has $n=ml$ linearly independent right latent vectors $p_{ij}$ ($i=1,\cdots,l$) and ($j=1,\cdots,m$) (left latent vectors $q_{ij}$ $i=1,\cdots,l$ and $j=1,\cdots,m$) corresponding to latent roots $\lambda_{ij}$ then  $P_i \Lambda_i P_i^{-1},(Q_i\Lambda_i Q_i^{-1})$   is a right (left) solvent. \\
 		Where:  $P_i=[p_{i1},p_{i2},...,p_{im}],      (Q_i=[q_{i1},q_{i2},...,q_{im}]^T )$    and $\Lambda_i=diag(\lambda_{i1},\lambda_{i2},\cdots,\lambda_{im})$.\\
 	\end{theorem}
 	
 	\begin{theorem} If $A(\lambda)$  has  $n=ml$  latent roots $\lambda_{i1},\lambda_{i2},\cdots,\lambda_{im}$ and the corresponding right latent vectors $p_{i1},p_{i2},...,p_{im}$ has as well as the left latent vectors $q_{i1},q_{i2},...,q_{im}$ are both linearly independent, then the associated right solvent $R_i$ and left solvent  $L_i$ are related by:$R_i=W_iL_iW_i^{-1}$ Where: $W_i=P_iQ_i$  and  $P_i=[p_{i1},p_{i2},...,p_{im}],     (Q_i=[q_{i1},q_{i2},...,q_{im}]^T$). and $" T "$ stands for transpose 
 	\end{theorem}
 	
 	For analysis and design of large-scale multivariable systems, it is necessary to determine a complete set of solvents of the matrix polynomial. Given the matrix polynomial $A(\lambda)$ if a right solvent $R$ is obtained, the left solvent of $L$ of $A(\lambda)$ associated with $R$ can be determined by using the following [\cite{24}]:
 	\begin{equation}
 	L=Q^{-1}RQ   \hspace{0.5cm}  \mbox{rank}(Q)=m
 	\end{equation}
 	Where $Q$ is the solution of the following linear matrix equation [\cite{24}]:
 	\begin{equation}
 	\sum \limits_{i=0}^{l-1}R^{l-1-i}QB_i=I_m
 	\end{equation}
 	Or in vector form using the Kronecker product we have
 	\begin{equation}
 	Vec(Q)=
 	\left(
 	\begin{array}{c}
 	\sum \limits_{i=0}^{l-1}{B_i}^T\bigotimes
 	\left(
 	\begin{array}{c}
 	R^{l-1-i}
 	\end{array}
 	\right)
 	\end{array}
 	\right)^{-1}
 	Vec(I_m)
 	\end{equation}
 	Where:$\bigotimes$ designates the Kronecker product, and $B_i$ are the matrix coefficients of $B(\lambda)$ with
 	$\lambda I_m-R$ factored out from $A(\lambda)$, i.e.,
 	\begin{equation}
 	B(\lambda)=A(\lambda)\left(
 	\begin{array}{c}
 	\lambda I_m-R
 	\end{array}
 	\right)^{-1}=\sum \limits_{i=0}^{l-1}B_i\lambda^{l-1-i}
 	\end{equation}
 	\begin{equation}
 	B(\lambda)=B_0\lambda^{l-1}+B_1\lambda^{l-2}+...+B_{l-1}
 	\end{equation}
 	We can compute the coefficients $B_i$, using the algorithm of synthetic division:
 	\vspace{-.6cm}
 	\begin{table}[H]
 		\hspace*{0.5cm}
 		\begin{center}
 			{\begin{tabular}{cc}
 					\hspace{1.5cm} $B_0=I_m$  & \hspace{0.5cm}\\
 					\hspace{1.5cm} $B_1=B_0A_1+B_0R$ & \hspace{0.5cm}\\
 					\hspace{1.5cm} $\cdots$ & \hspace{0.5cm}\\
 					\hspace{1.5cm} $B_k=B_0A_k+B_{k-1}R$ & \hspace{0.2cm}$k=1,2,...,l-1$\\
 					\hspace{1.5cm} $O_m=B_0A_l+B_{l-1}R$ & \hspace{0.5cm}
 			\end{tabular}}
 		\end{center}
 	\end{table}
 	\vspace{-.6cm}
 	\hspace{-.4cm}
 	
 	\begin{theorem}If the elementary divisors of $A(\lambda)$  are linear, then $A(\lambda)$  can be factored into the product of  $l$-linear monic $\lambda$-matrices called a complete set of spectral factors.
 		\begin{equation}
 		A(\lambda)=(\lambda I_m-Q_l )(\lambda I_m-Q_{l-1} )...(\lambda I_m-Q_1 )
 		\end{equation}
 		where: $(\lambda I_m-Q_i ),i=1...l$   are referred to as a complete set of linear spectral factors. The  $m\times m$ complex matrices $Q_i,i=1...l$ are called the spectral factors of the $\lambda$-matrix $A(\lambda)$.
 	\end{theorem}

 	The most right spectral factor $Q_1$ is a right solvent of $A(\lambda)$  and the most left spectral factor $Q_l$ is a left solvent of $A(\lambda)$, whereas the spectral factors may or may not be solvents of $A(\lambda)$. The relationship between solvents and spectral factors are studied by Shieh and Tsay in reference [24].

 	\subsection{Transformation of solvents to spectral factors}
 	
 	The diagonal forms of a complete set of solvents and those of a complete set of spectral factors are identical and are related by similarity transformation.
 	
 	\begin{theorem}:[\cite{19}] Consider a complete set of right solvents $\{R_1,R_2,...,R_l \}$ of monic$\lambda$-matrix  $A(\lambda)$ , then $A(\lambda)$ can be factored as:
 		\begin{equation*}
 		A(\lambda)=N_l(\lambda)=(\lambda I_m-Q_l )(\lambda I_m-Q_{l-1} )...(\lambda I_m-Q_1 )
 		\end{equation*}
 		By using the following recursive scheme: (for $k=1,...,l$)
 		\begin{equation}
 		Q_k=\left(
 		\begin{array}{c}
 		N_{(k-1)R}(R_k)
 		\end{array}
 		\right)
 		R_k
 		\left(
 		\begin{array}{c}
 		N_{(k-1)R}(R_k)
 		\end{array}
 		\right)^{-1}
 		\end{equation}
 		where:
 		\begin{equation}
 		N_{k}(\lambda)=(\lambda I_m - Q_k)N_{k-1}(\lambda) \hspace{1cm} k=1,...,l
 		\end{equation}
 		and for any $j$ we write
 		\begin{equation*}
 		N_{kR}(R_j)=N_{(k-1)R}(R_j)R_j-Q_kN_{(k-1)R}(R_j),\hspace{0.3cm} k=1,...,l
 		\end{equation*}
 		with:\\ 
 		$ N_{0}(\lambda)=I_m \hspace{.2cm} N_{0R}(R_j)=I_m \hspace{.2cm} \mbox{for any} \hspace{.2cm} j$ and \\
 		$\mbox{rank}(N_{(k-1)R}(R_k))=m, \hspace{.1cm} k=1,...,l$
 	\end{theorem}

 	Similarly the spectral factors can be obtained from the known $L_i$ of $A(\lambda)$ as follow: (for $k=1,...,l$)
 	\begin{equation}
 	Q_k=Q_{l+1-k}
 	\end{equation}
 	\begin{equation}
 	Q_k=\left(
 	\begin{array}{c}
 	M_{(k-1)L}(L_k)
 	\end{array}
 	\right)^{-1}
 	L_k
 	\left(
 	\begin{array}{c}
 	M_{(k-1)L}(L_k)
 	\end{array}
 	\right) 
 	\end{equation}
 	where:
 	\begin{equation}
 	M_{k}(\lambda)=M_{k-1}(\lambda)(\lambda I_m - Q_k) \hspace{1cm} k=1,...,l
 	\end{equation}
 	and for any $j$ we write
 	\begin{equation*}
 	M_{kL}(L_j)=L_j M_{(k-1)L}(L_j)-M_{(k-1)L}(L_j)Q_k, \hspace{.3cm} k=1,...,l
 	\end{equation*}
 	with:\\
 	$M_{0}(\lambda)=I_m \hspace{.5cm} M_{0L}(L_j)=I_m \hspace{.5cm} \mbox{for any} \hspace{.5cm} j $\\
 	$\hspace{.5cm} \mbox{rank}(M_{(k-1)L}(L_k))=m \hspace{.5cm} \hspace{.8cm} k=1,...,l $\\ [0.2cm]
 	$M_{(k-1)L}(L_i)$ is a left matrix polynomial of $M_{(k-1)}(\lambda)$ having $\lambda$ replaced by a left solvent $L_j$
 	the spectral factorization of $A(\lambda)$ becomes:
 	\begin{equation*}
 	A(\lambda)=M_l(\lambda)=(\lambda I_m-Q_1 )(\lambda I_m-Q_{2} )...(\lambda I_m-Q_l ) \vspace{0.2cm}
 	\end{equation*}

 	\subsection{Transformation of spectral factors to solvents}
 	Given a complete set of spectral factors of a $\lambda$-matrix $A(\lambda)$, then a corresponding complete set of right (left) solvents can be obtained. The transformation of spectral factors to right (left) solvents of a $\lambda$-matrix can be derived as follow [\cite{19}]:
 	
 	\begin{theorem}\label{theorem7} Given a monic $\lambda$-matrix with all elementary devisors being linear
 		\begin{equation*}
 		A(\lambda)=(\lambda I_m-Q_l )(\lambda I_m-Q_{l-1} )...(\lambda I_m-Q_1 )
 		\end{equation*}
 		where $Q_i$ ($\triangleq Q_{l+1-i}$) $i=1,...,l$ are a complete set of spectral factors of a $\lambda$-matrix $A(\lambda)$, and $Q_{i}\bigcap Q_{j}=\varnothing$\\ [.2cm]
 	\end{theorem}
 	
 	Now let us define $\lambda$-matrices $N_i(\lambda)$ $i=1,...,l$ as follow:
 	\begin{equation}
 	N_i(\lambda)=(\lambda I-Q_i)^{-1}N_{i-1}(\lambda)
 	\end{equation}
 	\begin{equation}
 	N_i(\lambda)=I_m\lambda^{l-i} + A_{1i}\lambda^{l-i-1}+...+A_{(l-i-1)i}\lambda+A_{(l-i)i} 
 	\end{equation}
 	with $N_0=A(\lambda)$ then the transformation matrix $P_i$ which transforms the spectral factor $Q_i$
 	($\triangleq Q_{l+1-i}$) to the right solvent $R_i$ ($\triangleq R_{l+1-i}$)of $A(\lambda)$ can be
 	constructed from the new algorithm as follow: $\hspace{.3cm} (\mbox{rank}(P_i)=m) \hspace{.3cm}$
 	\begin{equation}
 	R_i\triangleq R_{l+1-i}=P_iQ_i{P_i}^{-1}  i=1,...,l
 	\end{equation}
 	where: the $m\times m$ matrix $P_i$ can be solved from the following matrix equation $i=1,...,m$
 	\begin{equation}
 	Vec(P_i)=(G_{Ni}(Q_i))^{-1}Vec(I_m)
 	\end{equation}
 	where $G_{Ni}(Qi)$  $(\mbox{rank}(G_{N_i}(Q_i))=m^2)$ is defined by:\\ [.2cm]
 	\begin{equation*}
 	\begin{array}{cc}
 	G_{Ni}(Q_i)\triangleq & {(Q_i^{l-i})}^T{\bigotimes} I_m+{(Q_i^{l-i-1})}^T{\bigotimes}
 	A_{1i}+... \\
 	~~ & +{Q_i}^T{\bigotimes} A_{(l-i)i}+I_m{\bigotimes} A_{(l-i)i}
 	\end{array}
 	\end{equation*}

 	in the same way the complete set of spectral factors $Q_i, \hspace{0.5cm} i=1,2,...,l$ can be converted into the left solvent $L_i, \hspace{0.5cm} i=1,2,...,l$ using the following algorithm:
 	\begin{equation}
 	M_i(\lambda)=M_{i-1}(\lambda)(\lambda I_m-Q_i)^{-1}, \hspace{0.2cm} i=1,...,l
 	\end{equation}
 	\begin{equation}
 	M_i(\lambda)=I_m\lambda^{l-i} + A_{1i}\lambda^{l-i-1}+...+A_{(l-i-1)i}\lambda+A_{(l-i)i}
 	\end{equation}
 	\begin{equation*}
 	\begin{array}{cc}
 	H_{Mi}(Q_i)\triangleq & I_m{\bigotimes} Q_i^{l-i}+{A^T_{1i}}{\bigotimes}Q_i^{l-i-1}
 	+... \\
 	~~ & +{A^T_{(l-i-1)i}}{\bigotimes}Q_i +A^T_{(l-i)i}{\bigotimes} I_m 
 	\end{array}
 	\end{equation*}
 	\begin{equation*}
 	Vec(S_i)=(H_{Mi}(Q_i))^{-1}Vec(I_m) , \hspace{.5cm} \mbox{rank}(H_{M_i}(Q_i))=m^2
 	\end{equation*}
 	\begin{equation}
 	L_i=S_i^{-1}Q_iS_i \hspace{.5cm} i=1,...,l
 	\end{equation}
 	
 	\subsection{Block companion forms}
 	A useful tool for the analysis of matrix polynomials is the block companion form matrix.
 	Given a $\lambda$-matrix as in eq(\ref{eq.1}) where $A_i \in C^{m\times m}$ and $\lambda\in C$, the associated block companion
 	form matrices right (left) are:
 	\begin{equation}
 	{A_R=}{\left(
 		\begin{array}{cccc}
 		O_m & I_m & \cdots & O_m \\
 		O_m & O_m & \cdots & O_m \\
 		\vdots & \vdots & \ldots & O_m \\
 		O_m & O_m & \vdots & I_m \\
 		-A_l & -A_{l-1} & \cdots & -A_1 \\
 		\end{array}
 		\right)},
 	{A_L=}{\left(
 		\begin{array}{ccccc}
 		O_m  & \cdots & O_m & -A_{l} \\
 		I_m  & \cdots & O_m & -A_{l-1} \\
 		\vdots  & \vdots & \vdots & \vdots \\
 		O_m & \cdots & O_m & -A_{2} \\
 		O_m & \cdots & I_m & -A_{1} \\
 		\end{array}
 		\right)} \vspace{0.3cm}
 	\end{equation}
 	Note that: $A_L$  is the block transpose of $A_R$  . If the matrix polynomial $A(\lambda)$  has a complete set of solvents,
 	these companion matrices can be respectively block diagonalised via the right(left) block Vandermande matrix defined by:
 	
 	\begin{equation}
 	V_{R}=\begin{pmatrix}
 	&I_{m}  &I_{m}  &\cdots & I_{m}\\ 
 	&R_{1}  &R_{2} &\cdots  &R_{l} \\ 
 	&\vdots   &\vdots  &\ddots &\vdots \\ 
 	&R_{1}^{l-1}  &R_{2}^{l-1} &\cdots  &R_{l}^{l-1} 
 	\end{pmatrix}
 	V_{L}=\begin{pmatrix}
 	&I_{m}  &L_{1}  &\cdots & L_{1}^{l-1}\\ 
 	&I_{m}  &L_{2}  &\cdots & L_{2}^{l-1}\\ 
 	&\vdots   &\vdots  &\ddots  &\vdots \\ 
 	&I_{m}  &L_{l}  &\cdots & L_{l}^{l-1} 
 	\end{pmatrix}
 	\end{equation}
 	where $R_1,R_2,...,R_l$ and/or $L_1,L_2,...,L_l$ represent the complete set of right (left) solvents.
 	Since the block Vandermande matrices are nonsingular [\cite{1}],[\cite{2}] and [\cite{39}] we can write \vspace{0.2cm}
 	\begin{equation}
 	{V_R}^{-1}A_RV_R=\mbox{Blockdiag}(R_1,R_2,...,R_l) \vspace{0.2cm}
 	\end{equation}
 	\begin{equation}
 	{V_L}^{-1}A_LV_L=\mbox{Blockdiag}(L_1,L_2,...,L_l) \vspace{0.3cm}
 	\end{equation}
 	These similarity transformations do a block decoupling of the spectrum of $A(\lambda)$ which is very useful
 	in the analysis and design of large order control systems.

 \section{Special Fractorization Algorithms}
 
 In this section we are going to present algorithms that can factorize a linear term from
 a given matrix polynomial. Firstly we give the generalized quotient difference algorithm and next we give a new
 extended algorithm based on the Horner scheme. The matrix quotient-difference $Q.D.$ algorithm is a generalization of the scalar case [\cite{35}] and it is developed in [\cite{03}]
 The scalar $Q.D.$ algorithm is used for finding the roots of a scalar polynomial. The Quotient-Difference scheme for matrix polynomials can be defined just like the scalar one [\cite{09}] and it is consists on building a table that we call the $Q.D.$ tableau.
 
 \subsection{The right block matrix Q.D. algorithm}
 Given a matrix polynomial with nonsingular coefficients as in eq(\ref{eq.1}).The objective is to find the spectral
 factors of $A(\lambda)$ that will allow as write $A(\lambda)$ as a product of $n$-linear factors as in eq(\ref{eq.1}).
 The block left companion form, is:
 \begin{equation}
 C_3={\left(
 	\begin{array}{ccccc}
 	-A_{1} & I_m &\cdots&  O_m\\
 	-A_{2} & O_m &\cdots& O_m\\
 	\vdots  & \vdots &\vdots& \vdots \\
 	-A_{l-1} & \cdots &\cdots& I_m\\
 	-A_{l}& \cdots  &\cdots& O_m \\
 	\end{array}
 	\right)}
 \end{equation}
 The required transformation is a sequence of $LR$ decomposition such that:
 \begin{equation}
 C_3=\left(
 \begin{array}{cc}
 C_{11} & C_{12} \\
 C_{21} & C_{22} \\
 \end{array}
 \right)=\left(
 \begin{array}{cc}
 I_{m} & O_{m} \\
 X_{m} & I_{m} \\
 \end{array}
 \right)\left(
 \begin{array}{cc}
 A & B \\
 C & D \\
 \end{array}
 \right)
 \end{equation}
 where:
 \begin{equation*}
 C_{11}={\left(
 	\begin{array}{ccccc}
 	-A_{1} & I_m &\cdots&  O_m\\
 	-A_{2} & O_m &\cdots& O_m\\
 	\vdots  & \vdots &\vdots& \vdots \\
 	-A_{l-2} & \cdots &\cdots& I_m\\
 	-A_{l-1}& \cdots  &\cdots& O_m \\
 	\end{array}
 	\right)},
 C_{12}={\left(
 	\begin{array}{c}
 	O_m \\
 	O_m \\
 	\vdots \\
 	O_m \\
 	I_m \\
 	\end{array}
 	\right)}
 \end{equation*}
 \begin{equation*}
 C_{21}={[-A_l\hspace{0.2cm}O_m\hspace{0.2cm}O_m,...,O_m]}, \hspace{0.5cm} C_{22}={[O_m]}
 \end{equation*}
 It is required to have $C=0$, then let
 \begin{equation}
 X={[-X_1\hspace{0.2cm}X_2\hspace{0.2cm}X_3,...,X_{l-1}]}
 \end{equation}
 We obtain the following set of equations:
 \begin{table}[H]
 	\begin{center}
 		{\begin{tabular}{c}
 				$X_1A_1+X_2A_2+...+X_{l-1}A_{l-1}=A_l$\\
 				$X_1=X_2=...=X_{l-2}=0$\\
 				$X_{l-1}+D=0$ \vspace{-0.5cm}
 		\end{tabular}}
 	\end{center}
 \end{table}
 \hspace{-0.5cm}Leading the following decomposition of $C_3$:
 \begin{equation*}
 {C_3=\left(
 	\begin{array}{ccccc}
 	I_m  & \cdots & O_m & O_m \\
 	O_m  & \cdots & O_m & O_m \\
 	\vdots & \vdots & \vdots & \vdots \\
 	O_m  & \cdots & I_m & O_m \\
 	O_m & \cdots & A_lA^{-1}_{l-1} & I_m \\
 	\end{array}
 	\right)
 	\left(
 	\begin{array}{ccccc}
 	-A_{1} & I_m &\cdots&  O_m\\
 	-A_{2} & O_m &\cdots& O_m\\
 	\vdots  & \vdots &\vdots& \vdots \\
 	-A_{l-1} & O_{m} &\cdots& I_m\\
 	O_{m}& O_{m}  &\cdots& -A_lA^{-1}_{l-1} \\
 	\end{array}
 	\right)}
 \end{equation*}
 Hence $C_3$ can be written as:
 \begin{equation}
 C_3=L_{-(l-2)}R_{-(l-2)}
 \end{equation}
 Continuing this process of the block $C_{11}$ up when $C_3$ is equivalent to a matrix $R_0$
 \begin{equation}
 C_3=L_{-(l-2)}L_{-(l-3)}...L_0R_0
 \end{equation}
 \begin{equation}
 R_0={\left(
 	\begin{array}{ccccc}
 	-A_1 & I_m & \cdots & O_m & O_m \\
 	O_m & -A_2A^{-1}_1 & \cdots & O_m & O_m \\
 	\vdots & \vdots & \vdots & \vdots & \vdots \\
 	O_m & O_m & \cdots & -A_{l-1}A^{-1}_{l-2} & I_m \\
 	O_m & O_m & \cdots & O_m & -A_lA^{-1}_{l-1} \\
 	\end{array}
 	\right)}
 \end{equation}
 It is clear that if the matrices $L_0,L_{-1},...,L_{l-2}$ are identity matrices, then the block
 companion matrix $C_{3}$ will be :
 \begin{equation}
 M={\left(
 	\begin{array}{ccccc}
 	Q_1 & I_m & \cdots & O_m & O_m \\
 	O_m & Q_2 & \cdots & O_m & O_m \\
 	\vdots & \vdots & \vdots & \vdots & \vdots \\
 	O_m & O_m & \cdots & Q_{l-1} & I_m \\
 	O_m & O_m & \cdots & O_m & Q_{l} \\
 	\end{array}
 	\right)}
 \end{equation}
 The following theorem shows that under certain conditions, the sequence of  $L_0,L_{-1},...,L_{l-2}$  converges to
 identities see [\cite{37}] and [\cite{09}]:

 \begin{theorem} Let $M=X \Lambda X^{-1}$ where
 	\begin{equation}
 	\Lambda={\left(
 		\begin{array}{cccc}
 		R_1 & O_m & \cdots & O_m \\
 		O_m & R_2 & \cdots & O_m \\
 		\vdots & \vdots & \ddots & \vdots \\
 		O_m & O_m & \cdots & R_l \\
 		\end{array}
 		\right)}
 	\end{equation}
 	
 	If the following conditions are satisfied:\\ [.15cm]
 	a) dominance relation between $R_k$: $R_1>R_2>....R_l$ \\
 	b) $X^{-1}=Y$ has a Block $LR$ factorization $L_yR_y$ \\
 	c) $X$ has a Block $LR$ factorization $L_xR_x$ \\ [0.15cm]
 	Then the block $LR$ algorithm just defined converges (i.e. $L_k\rightarrow I$)\\ [0.15cm]
 	
 \end{theorem}

 The above theorem states that we can start the $Q.D.$ algorithm by considering:
 \begin{equation*}
 {
 	{E_1}^{(0)}=-A_2{A_1}^{-1}, \hspace{0.5cm} {E_2}^{(-1)}=-A_3{A_2}^{-1},}
 \end{equation*}
 \begin{equation*}
 {
 	{E_3}^{(-2)}=-A_4{A_3}^{-1},...,{E_{l-1}}^{-(l-2)}=-A_l{A_{l-1}}^{-1}}
 \end{equation*}
 \begin{equation*}
 {
 	{Q_1}^{(0)}=-A_1, {Q_2}^{(-1)}=O_m,...,{Q_{l-1}}^{-(l-2)}=O_m}
 \end{equation*}
 The  last two equations provide us with the first two rows of the $Q.D.$ tableau (one row of $Q's$ and one row of $E's$). Hence, we can solve the rhombus rules for the bottom element (called the south element by Henrici [\cite{35}]). We obtain the row generation of the $Q.D.$ algorithm:
 \begin{equation}
 {
 	\left\{
 	\begin{array}{ll}
 	{Q_i}^{(j+1)}={Q_i}^{(j)}+{E_i}^{(j)}-{E_i}^{(j+1)}  \\
 	{E_i}^{(j+1)}={Q_{i+1}}^{(j)}{E_i}^{(j)}\left[
 	\begin{array}{c}
 	{Q_i}^{(j+1)}
 	\end{array}
 	\right]^{-1}
 	\end{array}
 	\right.}
 \end{equation}
 Writing this in tabular form yield
 
 \begin{figure}[!h]
 	\centering
 	\includegraphics[scale=0.3]{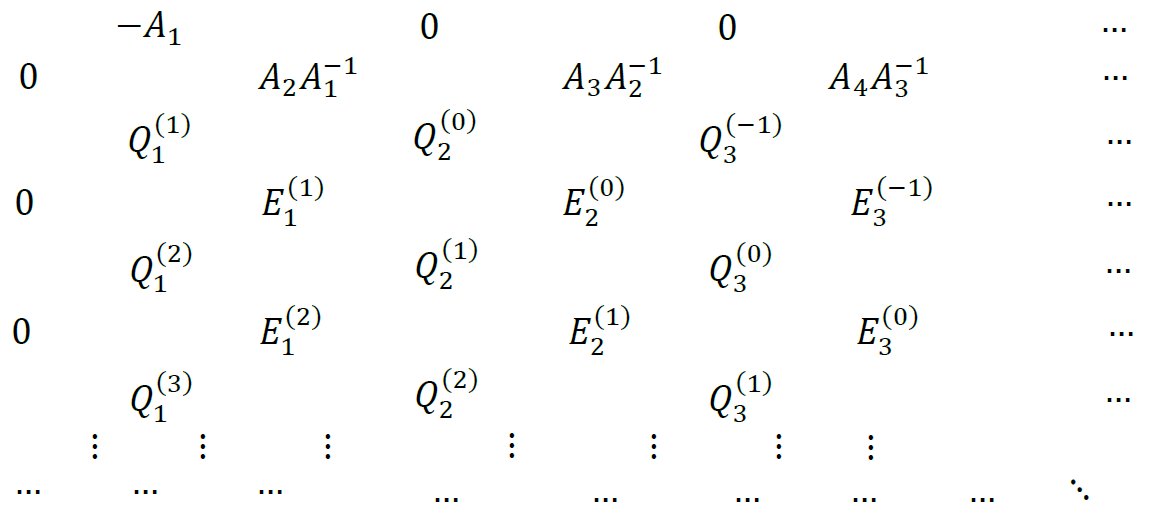}\label{figure1}
 	\caption*{$Q.D$ algorithm}
 \end{figure}

 where the ${Q_i}^{(j)}$ are the spectral factors of $ A(\lambda)$ . In addition, the $Q.D.$ algorithm gives all spectral factors simultaneously and in dominance order. We have chosen, in the above, the row generation algorithm because it is numerically stable see reference [\cite{09}] for details.

 \begin{example} Consider a matrix polynomial of $2^{nd}$ order and $3^{nd}$ degree with the following matrix coefficients.
 	\begin{tabbing}
 		$A_0={\left(
 			\begin{array}{cc}
 			1 & 0 \\
 			0 & 1 \\
 			\end{array}
 			\right)}$, \hspace{0.5cm}
 		$A_1={\left(
 			\begin{array}{cc}
 			-27.1525 & 0.8166 \\
 			-179.7826 & 38.1525 \\
 			\end{array}
 			\right)}$ \\ \\
 		$A_2={\left(
 			\begin{array}{cc}
 			116.4 & 85.0 \\
 			1043.4 & 836.7 \\
 			\end{array}
 			\right)}$, \hspace{0.5cm}
 		$A_3={\left(
 			\begin{array}{cc}
 			126.9 & 353.5 \\
 			1038.7 & 2947.6 \\
 			\end{array}
 			\right)}$
 	\end{tabbing}
 	We apply now the generalized row generation $Q.D.$ algorithm to find the complete set of spectral factors and then
 	we use the similarity transformations given by Shieh [\cite{24}] to obtain the complete set of solvents both left and right equations(19)-(26).\\ [.15cm]
 	\textbf{Step 1:} initialization of the program to start\\
 	Enter the degree and the order $m=2, l=3$ \\
 	Enter the number of iterations $N=35$\\
 	Enter the matrix polynomial coefficients $A_i$\\ [0.15cm]
 	\textbf{Step 2:} Construct $Q_1$ and $E_1$ the first row of $Q's$ and first row of $E's$\\
 	$Q_1=[-A_1A_0^{-1} \hspace{0.2cm} O_2 \hspace{0.2cm}O_2 ],\hspace{0.2cm}E_1=[O_2 \hspace{0.2cm} A_2 A_1^{-1}\hspace{0.2cm}A_3 A_2^{-1}\hspace{0.2cm}O_2]$\\ [.15cm]
 	\textbf{Step 3: } Building or generating the rest rows using the rhombus rules \vspace{-0.5cm} 
 	\begin{figure}[H]
 		\hspace{-.4cm}
 		\includegraphics[width=9cm,height=7cm]{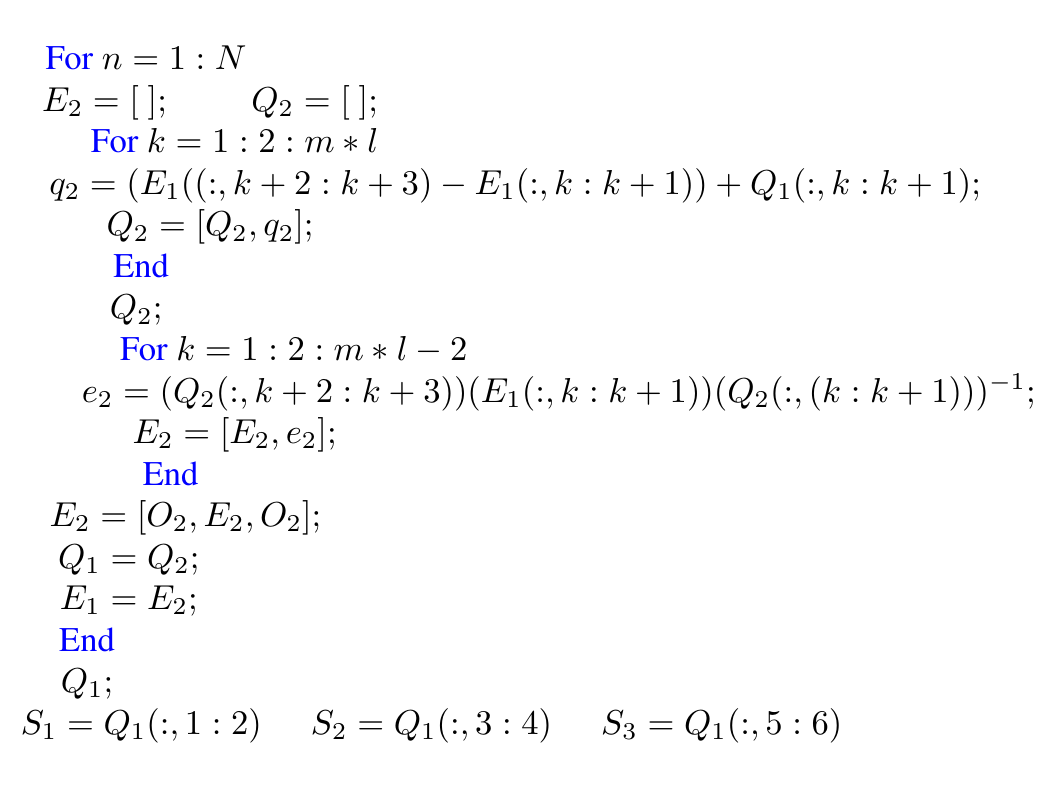}\vspace{-.6cm}
 	\end{figure}
 	
 	\hspace{-0.35cm}Running the above steps (1) to (3) we  obtain the following complete set of spectral factors $S_i$  :\\
 	\begin{equation*}
 	Q_1={\left(
 		\begin{array}{cccccc}
 		~~~~3.0000 & ~~~~2.0000 & -8.2908 & 0.7118 & 32.4434 & -3.5284 \\
 		-90.000 & -15.000 & -16.8400 & 8.1248 & 286.6226 & -31.2773
 		\end{array}
 		\right)}
 	\end{equation*}
 	\begin{equation*}
 	S_1={\left(
 		\begin{array}{cc}
 		~~~~3.0 & ~~~~2.0 \\
 		-90.0 & -15.0 \\
 		\end{array}
 		\right)},
 	S_2={\left(
 		\begin{array}{cc}
 		-8.2908 & 0.7118 \\
 		-16.8400 & 8.1248 \\
 		\end{array}
 		\right)},
 	\end{equation*}
 	\begin{equation*}
 	S_3={\left(
 		\begin{array}{cc}
 		32.4434 & -3.5284 \\
 		286.6226 & -31.2773 \\
 		\end{array}
 		\right)}
 	\end{equation*}
 	Now, we should extract a complete set of right solvent from those block spectra using the algorithmic similarity
 	transformations in equations from (21) to (24).\\ [.3cm]
 	\textbf{Step 4:} Reverse the orientation of spectral factors \\
 	$U_1=S_3; \hspace{.3cm}U_2=S_2; \hspace{.3cm}U_3=S_1$ \\
 	\textbf{Step 5:}Evaluate the coefficients  using the synthetic long division and then find the corresponding transformation matrix as in theorem \ref{theorem7}.
 	\begin{table}[H]
 		\begin{center}
 			{\begin{tabular}{c}
 					$N_{11}=A_1+U_1;$\\
 					$N_{12}=A_2+U_1\star N_{11};$\\
 					$G_1=(U_1^2 )^T \bigotimes I_2+(U_1 )^T\bigotimes N_{11}+I_2\bigotimes N_{12};$\\
 					$vecp_1=G_1^{-1} \star [1;0;0;1]$\\
 					$p_1=[vecp_1 (1:2),vecp_1 (3:4)];$\\
 					$R_1=p_1\star U_1\star (p_1)^{-1};$ \vspace*{-.3cm}
 			\end{tabular}}
 		\end{center}
 	\end{table}
 	\hspace{-0.45cm} You can verify the first solvent using: \\
 	rightzero1$=A_0\star(R_1)^3  + A_1\star(R_1)^2  + A_2\star R_1  + A_3$\\ [0.3cm]
 	\textbf{Step 6:} redo the same process for the next right solvents
 	\begin{table}[H]
 		\begin{center}
 			{\begin{tabular}{c}
 					$N_{21}=N_{11}+U_2;$ \\
 					$G_2=(U_2 )^T\bigotimes I_2+I_2\bigotimes N_{21};$ \\
 					$vecp_2=G_2^(-1)\star[1;0;0;1]$ \\
 					$p_2=[vecp_2 (1:2),vecp_2 (3:4)];$ \\
 					$R_2=p_2\star U_2\star (p_2)^{-1};$ \vspace*{-.3cm}
 			\end{tabular}}
 		\end{center}
 	\end{table}
 	\hspace{-0.45cm} For verification also you can use:\\
 	rightzero 2$=A_0\star (R_2)^3  + A_1\star (R_2)^2  + A_2\star R_2  + A_3$ \\ [0.3cm]
 	\textbf{Step 7:}The last solvents are obtained  directly from the most left spectral factor: $R_3=S_1$
 	or by usinguse the  transformation:
 	\begin{table}[H]
 		\begin{center}
 			{\begin{tabular}{c}
 					$G_3=(I_2 )^T\bigotimes I_2;$\\
 					$vecp_3=G_3^{-1}\star [1;0;0;1]$\\
 					$p_3=[vecp_3 (1:2),vecp_3 (3:4)  ];$\\
 					$R_3=p_3\star U_3\star p_3^{-1}=U_3;$\\ \vspace*{-.3cm}
 			\end{tabular}}
 		\end{center}
 	\end{table}
 	\hspace*{-.45cm}The final results are :
 	\begin{equation*}
 	R_1={\left(
 		\begin{array}{cc}
 		0.3637 & -4.5495 \\
 		-0.8183 & 0.8024 \\
 		\end{array}
 		\right)},
 	R_2={\left(
 		\begin{array}{cc}
 		7.2354 & 1.4024 \\
 		1.2995 & -7.4015 \\
 		\end{array}
 		\right)},
 	\end{equation*}
 	\begin{equation*}
 	R_3={\left(
 		\begin{array}{cc}
 		~~~~3.0000& ~~~~2.0000 \\
 		-90.000 & -15.000 \\
 		\end{array}
 		\right)}
 	\end{equation*}
 	Finally, we can also obtain the corresponding complete set of left solvents using the algorithmic similarity
 	transformation described in equations from (10) to (12).\\ [.3cm]
 	\textbf{Step 8:} coefficients determination using the synthetic long division
 	\begin{table}[H]
 		\begin{center}
 			{\begin{tabular}{c}
 					$B_0i=I_2;$ \\
 					$B_{1i}=A_1+R_i;$\\
 					$B_{2i}=A_2+B_{1i}\star R_i;$ \vspace*{-.5cm}
 			\end{tabular}}
 		\end{center}
 	\end{table}
 	\hspace*{-.35cm}
 	\textbf{Step 9:} find the corresponding similarity transformation matrix as in equations from equations (10) to (12). \vspace{-0.3cm}
 	\begin{figure}[H]
 		\includegraphics[width=9cm,height=3.5cm]{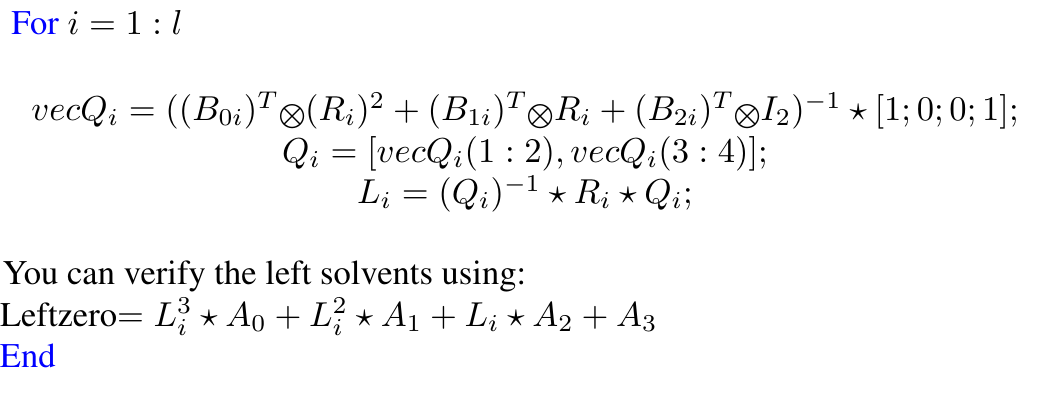}\vspace{-0.5cm}
 	\end{figure}
 	\hspace{-0.7cm}The left solvents are now obtained:\\
 	\begin{equation*}
 	L_1={\left(
 		\begin{array}{cc}
 		32.443 & -3.5284 \\
 		286.622 & -31.2773 \\
 		\end{array}
 		\right)},~~
 	L_2={\left(
 		\begin{array}{cc}
 		25.1323 & -2.8370 \\
 		204.5931 & -25.2983 \\
 		\end{array}
 		\right)},
 	\end{equation*}
 	
 	\begin{equation*}
 	L_3={\left(
 		\begin{array}{cc}
 		21.0123& -4.6531 \\
 		178.0910 & -33.0123 \\
 		\end{array}
 		\right)}
 	\end{equation*}
 \end{example}
 
 \subsection{Extended Horner algorithm}
 Horner's method is a technique to evaluate polynomials quickly. It needs $l$ multiplications and $l$ additions and  it is also a nested algorithmic programming that can decompose a polynomial into a
 multiplication of $l$ linear factors (Horner's method) based on the Euclidian synthetic long division. 
 
 Similarly Horner's method is a nesting technique requiring only $l$ multiplications and $l$
 additions to evaluate an arbitrary $l^{th}$-degree polynomial [\cite{33}]. 
 
 \begin{theorem} Let the function $P(x)$ be the polynomial of degree $l$ defined on the real field  $P: R\rightarrow R$
 	where: $a_i$  are constant coefficients and $x$ is real variable.
 	\begin{equation}
 	P(x)=a_0 x^l+a_1 x^{l-1}+...+a_{l-1} x+a_l
 	\end{equation}
 	If  $b_0=a_0$  and   $b_k=a_k+b_{k-1}\alpha ,   \hspace{.5cm} k=l,...,2,1$\\
 	Then  $b_l=P(\alpha) \hspace{.2cm} \mbox{and} \hspace{.2cm} P(x)$   can be written as:
 	\begin{equation}
 	P(x)=(x-\alpha)Q(x)+b_l
 	\end{equation}
 	Where:
 	\begin{equation}
 	Q(x)=b_0 x^{l-1}+b_1 x^{l-2}+...+b_{l-2} x+b_{l-1}\\ \vspace{0.2cm}
 	\end{equation}
 \end{theorem}
 
 \begin{proof}
 	The theorem can be proved using a direct calculation.\vspace{-0.1cm}
 	\begin{table}[H]
 		\begin{center}
 			{\begin{tabular}{c}
 					$P(x)=a_0 x^l+a_1 x^{l-1}+...+a_{l-1} x+a_l$\\ \\
 					$P(x)=(x-\alpha)(b_0 x^{l-1}+b_1 x^{l-2}+...+b_{l-2} x+b_{l-1} )+b_l$ \vspace{-.5cm}
 			\end{tabular}}
 		\end{center}
 	\end{table}
 	\hspace{-.4cm}Identifying the coefficients of $x$  with different powers we get:
 	\begin{table}[H]
 		\begin{center}
 			\vspace{-.2cm}
 			{\begin{tabular}{cc}
 					\hspace{1.5cm} $b_0=a_0$ & \hspace{0.5cm}\\
 					\hspace{1.5cm} $b_1=a_1+b_0\alpha$ & \hspace{0.5cm}\\
 					\hspace{1.5cm} $\vdots$ & \hspace{0.5cm}\\
 					\hspace{1.5cm} $b_k=a_k+b_{k-1}\alpha$  & where  $k=l,l-1,...,2$ \vspace{-.2cm}
 			\end{tabular}}
 		\end{center}\vspace{-0.4cm}
 	\end{table}
 	\hspace{-.4cm}Now if $\alpha$ is a root of the polynomial $P(x)$   , then $b_l$  should be zero, and  $a_l+b_{l-1}\alpha=0$ .\\
 	Hence, we may write
 	\begin{equation*}
 	\alpha=-\left(
 	\begin{array}{c}
 	\displaystyle{\frac{a_l}{b_{l-1}}}
 	\end{array}
 	\right)
 	\hspace{.3cm}  \mbox{or}  \hspace{.3cm}    x_{k+1}=-\left(\begin{array}{c}
 	\displaystyle{\frac{a_l}{b_{l-1,k}}}
 	\end{array}
 	\right)
 	\hspace{.3cm} k=0,1,...
 	\end{equation*}
 	The algorithm of Horner method in its recursive formula is then:
 	\begin{equation*}
 	b_{i,k}=a_k+b_{i,k-1} x_k;  \hspace{.3cm} i=1,...,l  \hspace{.3cm} \mbox{and} \hspace{.3cm}  b_{0,k}=a_0
 	\end{equation*}
 \end{proof}
 
 Now we generalize this nested algorithm to matrix polynomials, consider the monic $\lambda$-matrix
 $A(\lambda)$  and according to theorem \ref{theorem1} the matrix $A(\lambda)$  can be factored as: \\
 \begin{equation}
 A(\lambda)=Q(\lambda)(\lambda I-X)+A_R (X)
 \end{equation}
 where
 \begin{flalign*}
 A(\lambda)=&~~A_{0}\lambda^{l}+A_{1}\lambda^{l-1}+...+A_{l}=\sum \limits_{i=0}^{l} A_i \lambda^{l-i} \\
 Q(\lambda)=&~~B_{0}\lambda^{l-1}+B_{1}\lambda^{l-2}+...+B_{l-1}=\sum \limits_{i=0}^{l-1}B_i \lambda^{l-i-1} \\
 ~~&\hspace*{1.2cm} A_R(X)=\mbox{cst} 
 \end{flalign*}
 Using the algorithm of synthetic long division for matrices we obtain:
 \vspace*{-.6cm}
 \begin{table}[H]
 	\hspace*{0.5cm}
 	\begin{center}
 		{\begin{tabular}{cc}
 				\hspace{1.5cm} $B_0=A_0=I_m$  & \hspace{0.2cm}\\
 				\hspace{1.5cm} $B_1=B_0A_1+B_0X$ & \hspace{0.5cm}\\
 				\hspace{1.5cm} $\cdots$ & \hspace{0.5cm}\\
 				\hspace{1.5cm} $B_k=B_0A_k+B_{k-1}X$ & \hspace{0.2cm}$k=1,2,...,l-1$\\
 				\hspace{1.5cm} $O_m=B_0A_l+B_{l-1}X$ & \hspace{0.5cm}
 		\end{tabular}}
 	\end{center}\vspace{-0.6cm}
 \end{table}
 \hspace{-.45cm} From the last two equations we can iterate the process to get recursive algorithm as follow:\\ [.3cm]
 \textbf{Algorithm:}\vspace{-.4cm}
 \begin{figure}[H]
 	\includegraphics[width=6.5cm,height=5cm]{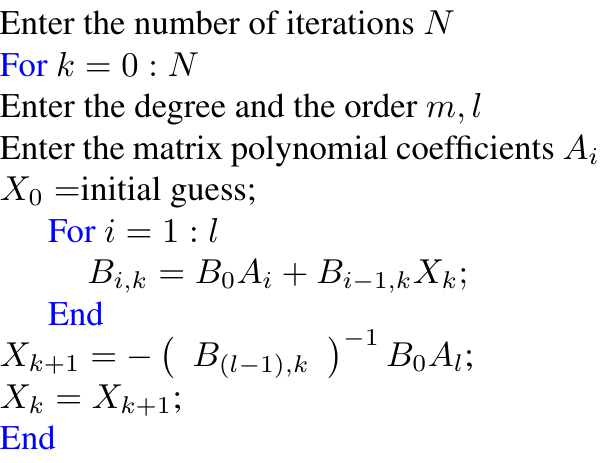}\vspace{-.5cm}
 \end{figure}
 \hspace{-.5cm}When you get the first spectral factor repeat the process until you get the complete set.
 
 \begin{example}Consider a matrix polynomial of $2^{nd}$ order and $3^{rd}$ degree with the following matrix
 	coefficients.
 	\begin{equation*}
 	A(\lambda)=A_{0}\lambda^{3}+A_{1}\lambda^{2}+A_{2}\lambda +A_{3}
 	\end{equation*}
 	With
 	\begin{equation*}
 	\hspace{.5cm} A_0={\left(
 		\begin{array}{cc}
 		1 & 0 \\
 		0 & 1 \\
 		\end{array}
 		\right)},\hspace{1.3cm}
 	A_1={\left(
 		\begin{array}{cc}
 		11.0000 & -1.0000 \\
 		6.7196 & 17.0000 \\
 		\end{array}
 		\right)},
 	\end{equation*}
 	\begin{equation*}
 	A_2={\left(
 		\begin{array}{cc}
 		30.0000& -11.0000 \\
 		70.9107 & 82.5304 \\
 		\end{array}
 		\right)}, \hspace{0.3cm}
 	A_3={\left(
 		\begin{array}{cc}
 		-0.0000 & -30.0000\\
 		182.0000 & 89.8393 \\
 		\end{array}
 		\right)}
 	\end{equation*}
 	We apply now the extended Horner's method via its algorithmic version to find the complete set of spectral factors and
 	then we use the similarity transformations given in [\cite{19}] to obtain the complete set of left and right  solvents \\ [.3cm]
 	The Block Horner scheme is:\vspace{-.05cm}
 	\begin{figure}[H]
 		\includegraphics[width=9cm,height=10cm]{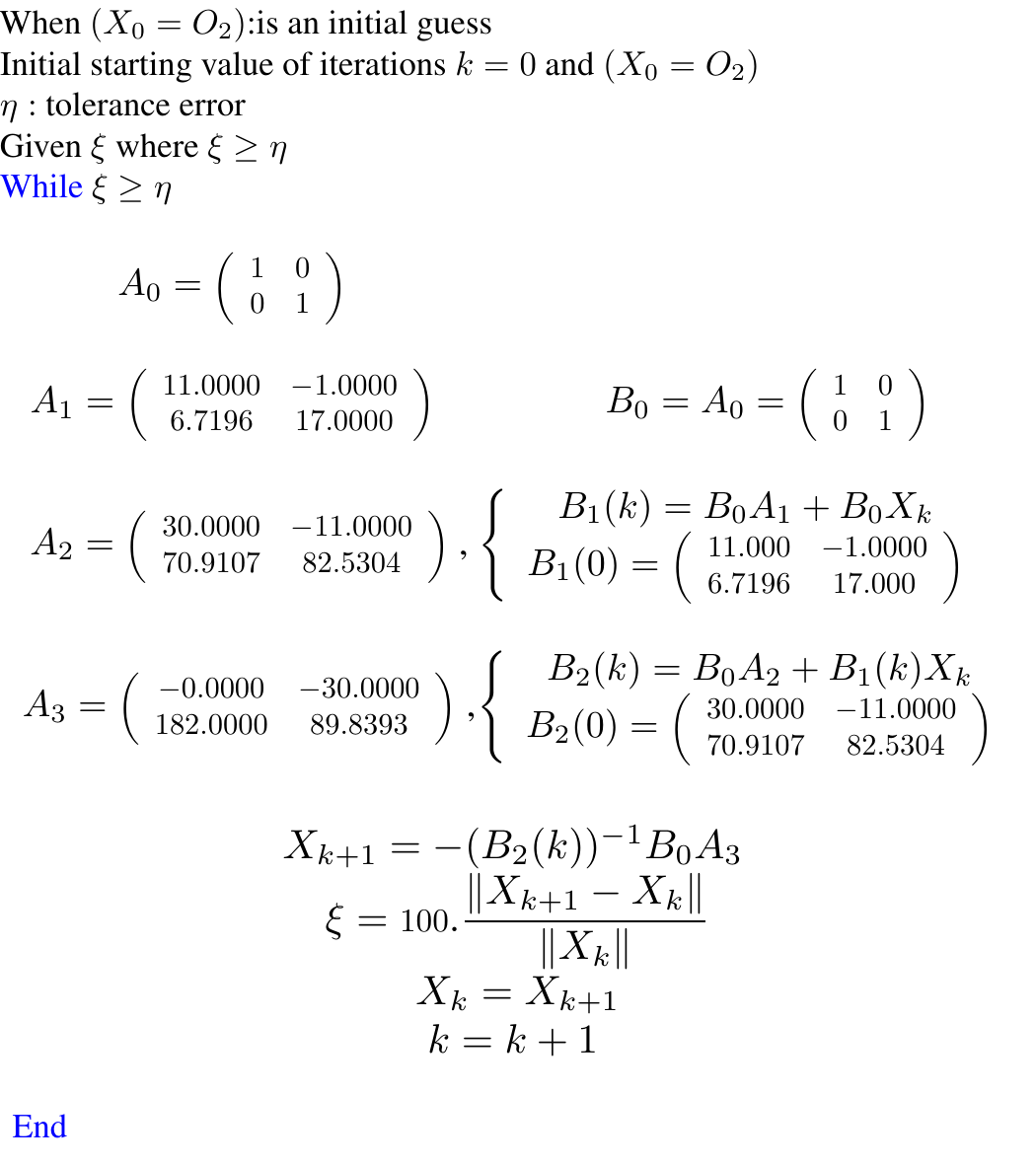}\vspace{-.6cm}
 	\end{figure}
 	\hspace{-.45cm}
 	Running the above algorithm we obtain the next complete set of spectral factors:
 	\begin{equation*}
 	S_1={\left(
 		\begin{array}{cc}
 		0.00 & 1.00 \\
 		-3.25 & 2.00 \\
 		\end{array}
 		\right)}, \hspace{0.5cm}
 	S_2={\left(
 		\begin{array}{cc}
 		-5.000 & 0.000 \\
 		-1.6042& -7.000 \\
 		\end{array}
 		\right)}
 	\end{equation*}
 	\begin{equation*}
 	S_3={\left(
 		\begin{array}{cc}
 		-6.000 & -0.000 \\
 		-1.8655 & 8.000 \\
 		\end{array}
 		\right)} \vspace{0.3cm}
 	\end{equation*}
 	Finally when we apply the similarity transformation algorithm as in equations from (21) to (24)
 	to right (or left) solvent form we get:
 	\begin{equation*}
 	R_1={\left(
 		\begin{array}{cc}
 		-5.9574 & 0.2553 \\
 		-0.3404 & -8.0426 \\
 		\end{array}
 		\right)},\hspace{.5cm}
 	R_2={\left(
 		\begin{array}{cc}
 		-4.9412 & 0.2941 \\
 		-0.4118 & -7.0588\\
 		\end{array}
 		\right)}
 	\end{equation*}
 	\begin{equation*}
 	R_3={\left(
 		\begin{array}{cc}
 		0.000 & 1 \\
 		-3.25 & -2 \\
 		\end{array}
 		\right)}
 	\end{equation*}
 \end{example}

 \subsubsection{Reformulation of the Block Horner method}
 An alternative form of the previous algorithm (under matrix and algebraic manipulations is:
 \begin{table}[H]
 	\begin{center}
 		{\begin{tabular}{cccc}
 				$B_0(k)=A_0=I_m$\\
 				$B_1(k)=B_0(k) A_1+B_0 (k)X(k)$\\
 				$...$\\
 				$B_{l-1}(k)=B_0(k) A_{l-1}+B_{l-2}(k)X(k)$\\
 				$O_m=B_0 (k) A_l+B_{l-1} (k)X(k)$ \vspace{-0.3cm}
 		\end{tabular}}
 	\end{center}
 \end{table}
 \begin{table}[H]
 	\begin{center}
 		{\begin{tabular}{cc}
 				$B_{l-1}(k)=A_{l-1}+...+A_1 X^{l-1}(k)+B_0 X^l (k)$\\
 				$\Rightarrow B_{l-1}(k)=[A_R (X_k )-A_l] X_k^{-1}$ \hspace{0.75cm} \\
 				$\Rightarrow (B_{l-1}(k))^{-1}=X_k[A_R (X_k)-A_l ]^{-1}$ \vspace{-0.3cm}
 		\end{tabular}}
 	\end{center}
 \end{table}
 \hspace{-0.5cm} Finally we obtain the following iterative formula: $(k=0,1,...)$
 \begin{equation}
 X_{k+1}=-(B_{l-1}(k))^{-1} A_l=X_k [A_l-A_R (X_k )]^{-1} A_l \\
 \end{equation}
 
 \begin{algorithm} 
 	Enter the degree and the order  $m,l$ \\
 	Enter the matrix polynomial coefficients $A_i\in R^{m\times m}$ \\
 	$X_0\in R^{m\times m}=$ initial guess; \\
 	Give some small $\eta$ and ($\delta =$initial start)$>\eta$ \\
 	$k=0$ \\ [0.35cm]
 	\textcolor[rgb]{0.00,0.00,1.00}{While}$\delta \geq \eta$

 	\begin{table}[H]
 		\begin{center}
 			{\begin{tabular}{cccc}
 					$X_{k+1}=X_k [A_l-A_R{X_k }]^{-1}A_l  ;$\\
 					$\delta =100.\displaystyle{\frac{\|X_{k+1}-X_k\|}{\| X_k\|}} ;$\\
 					$X_k\leftarrow X_{k+1};$\\
 					$k=k+1;$\\
 			\end{tabular}}\vspace{-0.4cm}
 		\end{center}
 	\end{table}
 	
 \end{algorithm}
 
 \textbf{Convergence condition:} Using  equations (44) we obtain the conditions for the the algorithm to converge to the  solution. \\ [0.3cm]
 1. \emph{Upper bound}\\
 \begin{equation*}
 \mbox{eq(45)} \Leftrightarrow X_{k+1}-X_k=X_{k+1} A_l^{-1} A_R (X_k ) \hspace{1.5cm}
 \end{equation*}
 \begin{equation*}
 \mbox{eq(45)} \Leftrightarrow \|X_{k+1}-X_k\|= \|X_{k+1} A_l^{-1} A_R (X_k )\| \hspace{.75cm}
 \end{equation*}
 \begin{equation*}
 \mbox{eq(45)} \Leftrightarrow \|X_{k+1}-X_k \|\leq \|X_{k+1}\|.\|A_l^{-1}\|.\|A_R {X_k}\|
 \end{equation*}
 \begin{equation*}
 \mbox{eq(45)} \Leftrightarrow \frac{\|X_{k+1}-X_k \|}{\left(
 	\begin{array}{c}
 	\|X_{k+1}\|.\|A_l^{-1} \|
 	\end{array}
 	\right)}\leq \|A_R (X_k)\| \hspace{1cm}
 \end{equation*}
 Now if $X_k$  tends to constant matrix $\|X_k \|\rightarrow M$ as $k \rightarrow \infty$ and $\|X_k^{-1} \|\rightarrow N$
 with $\|A_l^{-1}\|=\gamma$ and $\|A_l \|=\delta$  then:\\
 \begin{equation*}
 \lim_{k\rightarrow \infty}\|A_R (X_k )\|\geq \frac{\|X_(k+1)-X_k \|}{\left(
 	\begin{array}{c}
 	\|X_(k+1)\|.\|A_l^{-1}\|
 	\end{array}
 	\right)}
 \end{equation*}
 \begin{equation}
 \Rightarrow
 \lim_{k\rightarrow \infty}\|A_R (X_k )\| \geq \frac{\xi_k}{\gamma.M}\\
 \end{equation}
 
 2. \emph{Lower bound} \\
 \begin{equation*}
 A_R(X_k)=A_l[I-X_{k+1}^{-1}X_k]=A_l (X_{k+1})^{-1}[X_{k+1}-X_k ]
 \end{equation*}
 \begin{equation*}
 \Rightarrow \|A_R (X_k )\|\leq \|A_l \|.\|(X_{k+1})^{-1}\|.\|X_{k+1}-X_k\|
 \end{equation*}
 \begin{equation}
 \Rightarrow \lim_{k\rightarrow \infty}\|A_R (X_k )\| \leq \delta.N\xi_k\\
 \end{equation}
 From eq. (45) and (46) we obtain:
 \begin{equation}
 \frac{\xi_k}{\gamma.M} \leq \lim_{k\rightarrow \infty}\|A_R (X_k )\| \leq \delta.N\xi_k
 \end{equation}
 Finally if the matrix $X_k$  tends to constant matrix $X_k\rightarrow S$ and $(A_l-A_R (X_k))$  is nonsingular matrix then $S$
 is a solvent of the matrix polynomial $A_R (S)=O_m$.\\ [.5cm]
 \textbf{Convergence Type:} To get the convergence type we should evaluate a ratio relationship between any two successive differences.
 \begin{equation}
 X_{k+1}-S=X_k ([A_l-A_R (X_k )]^{-1} A_l-I)+X_k-S
 \end{equation}
 
 We define $F(X_k )=([A_l-A_R (X_k )]^{-1} A_l-I)$   then we have:
 \begin{equation}
 {\|X_k-S\|-\|X_k F(X_k )\|\leq \|X_{k+1}-S\| \leq \|X_k-S\|+\|X_k F(X_k )\|}
 \end{equation}
 We know that:
 \begin{equation}
 \lim_{k\rightarrow \infty}\|X_k F(X_k )\|=\lim_{k\rightarrow \infty} \Delta_k=\xi
 \end{equation}
 from equations (49) and (50) we deduce that:
 \begin{equation*}
 1-\frac{\Delta_k}{\|X_k-S\|} \leq \frac{\| X_{k+1}-S\|}{\|X_k-S\|} \leq 1+\frac{\Delta_k}{\|X_k-S\|}
 \end{equation*}
 Finally:\\
 \begin{equation}
 \lim_{k\rightarrow\infty}\frac{\| X_{k+1}-S\|}{\|X_k-S\|}=1\\ \vspace{0.3cm}
 \end{equation}
 
 \begin{example} consider the following matrix polynomial with repeated spectral factor:
 	\begin{equation*}
 	A(\lambda)=\left[
 	\begin{array}{c}
 	\lambda I-{\left(
 		\begin{array}{cc}
 		-7.1230 & -6.3246 \\
 		5.9279 & 5.1230\\
 		\end{array}
 		\right)}
 	\end{array}
 	\right]
 	^2=A_0 \lambda^2+A_1 \lambda+A_2 \hspace{0.3cm}
 	\end{equation*}
 	With
 	\begin{equation*}
 	A_0={\left(
 		\begin{array}{cc}
 		1 & 0 \\
 		0 & 1\\
 		\end{array}
 		\right)}, \hspace{0.8cm}
 	A_1={\left(
 		\begin{array}{cc}
 		14.2461 & 12.6493 \\
 		-11.8557 & -10.2461\\
 		\end{array}
 		\right)}
 	\end{equation*}
 	\begin{equation*}
 	A_2={\left(
 		\begin{array}{cc}
 		13.2461 & 12.6493 \\
 		-11.8557 & -11.2461\\
 		\end{array}
 		\right)}\\ \vspace{.3cm}
 	\end{equation*}
 	Find $X$ such that $A_R (X)=O_2$
 	\begin{equation*}
 	A_R (X)=A_0 X^2+A_1 X+A_2
 	\end{equation*}
 	If we apply the Block Horner algorithm we find \vspace{0.2cm}
 	\begin{equation*}
 	X_1={\left(
 		\begin{array}{cc}
 		-2.7323 & -1.8068 \\
 		1.6798 & 0.7521\\
 		\end{array}
 		\right)}, \hspace{0.8cm}
 	X_2={\left(
 		\begin{array}{cc}
 		-11.5138 & -10.8424 \\
 		10.1759 & 9.4939\\
 		\end{array}
 		\right)}\vspace{0.2cm}
 	\end{equation*}
 \end{example}

 \begin{remark} The Proposed Horner algorithm finds the whole set of spectral factors if it exists, even if there is no dominance factor among them.
 \end{remark}

 \subsubsection{Crossbred Newton$\_$Horner method}
 In order to accelerate the Block Horner method we make a crossbred (Hybrid)  Generalized Newton algorithm which is very fast due to its restricted local nature (i.e. Quadratic convergence). 
 
 \begin{figure}[H]
 	\centering
 	\includegraphics[width=7cm,height=.9cm]{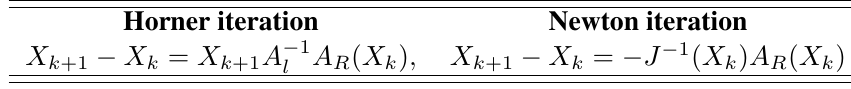}\vspace{-.2cm}
 \end{figure}
 
 \hspace{-0.6cm} Combine them we get:
 \begin{equation}
 X_{(k+1}=X_k+(X_k-J^{-1}(X_k)A_R (X_k)) {A_l}^{-1}A_R(X_k)
 \end{equation}
 where $J(X_k)$ is the Frechet differential.

 \begin{definition}
 	Let $B_1$ and  $B_2$ be Banach spaces and $A_R$ a nonlinear operator from $B_1$to $B_2$. If there exists a linear operator
 	$L$  from $B_1$ to $B_2$  such that:
 	\begin{tabbing}
 		\begin{tabular}{c}
 			$B_1\rightarrow B_2$\\
 			$H \rightarrow L(X+H)$\\
 			Where:\hspace*{8cm} \\
 			$\|A_R (X+H)-A_R (X)-L(X+H)\|=O(\|H\|)$\\
 			$X,H\in B_1 \hspace{0.8cm}A_R (X),L(X+H)\in B_2$\\
 		\end{tabular}
 	\end{tabbing}
 	Then $L(X+H)$ is called the Frechet derivative of $A_R$ at $X$ and sometimes is written  $dA_R (X,H)$. Also is read the
 	Frechet derivative of  $A_R$ at $X$ in the direction $H$. And  $J(X_k ).H=L(X+H)$   
 \end{definition}

 Algorithm\\
 Begin
 \begin{figure}[H]
 	\includegraphics[width=8.8cm,height=7cm]{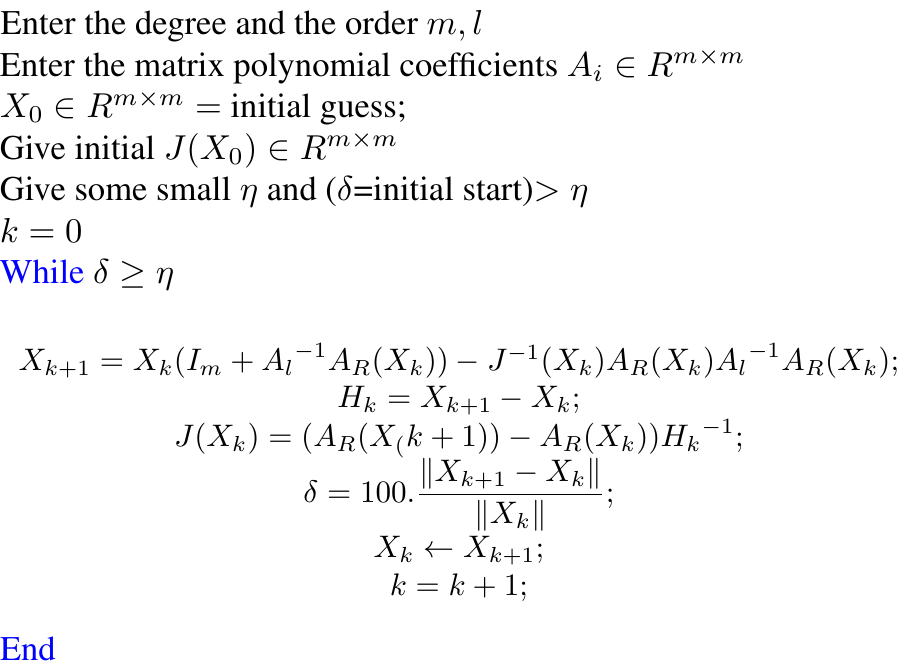}\vspace{-.3cm}
 \end{figure}

 \subsubsection{Two stage Block Horner algorithm}
 To accelerate the block Horner algorithm we use now a two stage Newton like iteration. Now by using
 theorem \ref{theorem1} we obtain:\vspace{-0.1cm}
 \begin{table}[H]
 	\centering
 	{\begin{tabular}{c}
 			$A_R(X)=(X-\Theta)(X^{l-1}+B_1X^{l-2}+...+B_{l-1})+B_l$\\
 			$\Rightarrow A_R (X)=(X-\Theta)Q(X)+B_l$\\
 			$\Rightarrow X-\Theta=(A_R(X)-B_l) Q^{-1}(X)$ \vspace{-0.15cm}
 	\end{tabular}}
 \end{table}
 \hspace*{-0.6cm}where $Q(X)=X^{l-1}+B_1 X^{l-2}+...+B_{l-1}$     and  $A_R (\Theta)=B_l$ Now if $\Theta$ is a solvent then $B_l=O_m$
 If now assume that $\Theta=X_{k+1}$  is a solvent to the matrix polynomial $A_R$  then $A_R (X_{k+1})=O_m$
 and $X_{k+1}=X_k-A_R(X_k) Q^{-1}(X_{k+1})$. Set also $Q(X)=(X-\Theta)(X^{l-1}+C_1 X^{l-2}+...+C_{l-2})+C_{l-1}$
 From the Horner scheme we can evaluate both $B_i$ and $C_i$  recursively:
 \begin{table}[H]
 	\hspace*{-0.25cm}
 	{\begin{tabular}{cc}
 			$B_0=I_m$ & \hspace{0.2cm} \\
 			$B_1=A_1+B_0 X$ & $C_0=I_m$\\
 			$B_2=A_2+B_1 X$ & $C_1=B_1+C_0 X$\\
 			$...$ &  $C_2=B_2+C_1 X$\\
 			$B_{l-1}=A_{l-1}+B_{l-2} X$ & $...$\\
 			$O_m =B_l=A_l+B_{l-1}X=A_R(X)$ &  $C_{l-1}=B_{l-1}+C_{l-2} X=Q(X)$\\ \\
 			\textbf{After iterating the last equation we get}: & \textbf{After iterating the last equation we get}:\\ \\
 			$B_l(k)=A_l+B_{l-1}(k)X_k$ &  $C_{l-1}(k)=B_{l-1}(k)+C_{l-2}(k)X_k$\\
 			$B_l(k)=A_R (X_k)$ &  $C_{l-1}(k)=Q(X_k)$\\
 	\end{tabular}}
 \end{table}


 \hspace{-.4cm}\textbf{Algorithm:} \vspace{0cm}
 \begin{figure}[H]
 	\includegraphics[width=8cm,height=8.4cm]{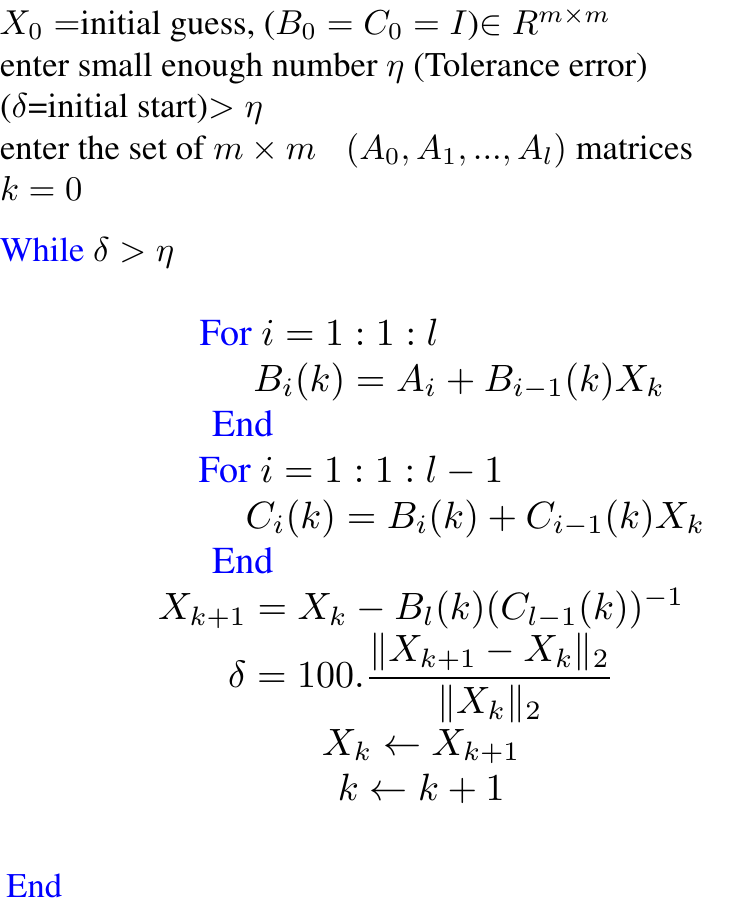}\vspace{-.2cm}
 \end{figure}

 \begin{remark} this two stage algorithm gathers the two advantages of Horner sachem and Newton algorithm because it is
 	nested programed nature, large sense independence on initial conditions and faster in execution due to the likeness
 	or the conformity to Newton method. 
 \end{remark}
 
 \begin{example} Given the following matrix polynomial
 	\begin{equation*}
 	A_R (X)=A_0 X^3+A_1 X^2+A_2 X+A_3
 	\end{equation*}
 	Where:
 	\begin{equation*}
 	A_0={\left(
 		\begin{array}{cc}
 		1 & 0 \\
 		0 & 1 \\
 		\end{array}
 		\right)}, \hspace{0.5cm}
 	A_1={\left(
 		\begin{array}{cc}
 		12.8793 & -0.4881 \\
 		-2.0989 & 15.1207 \\
 		\end{array}
 		\right)}\vspace{0.2cm}
 	\end{equation*}
 	\begin{equation*}
 	A_2={\left(
 		\begin{array}{cc}
 		56.5645 & -8.7887 \\
 		10.2686 & 55.9659\\
 		\end{array}
 		\right)}, \hspace{0.5cm}
 	A_3={\left(
 		\begin{array}{cc}
 		95.9331 & -37.5549 \\
 		160.9539 & -6.0938 \\
 		\end{array}
 		\right)}\vspace{0.2cm}
 	\end{equation*}
 	We apply the two stage Horner algorithm and after 15 iterations we get: \\ [0.2cm]
 	\begin{equation*}
 	X_0={\left(
 		\begin{array}{cc}
 		5.2114 & 4.8890 \\
 		2.3159 & 6.2406\\
 		\end{array}
 		\right)}, \hspace{0.5cm}
 	X_1={\left(
 		\begin{array}{cc}
 		-3.0729 & 1.4058 \\
 		-4.6569 & 1.0730\\
 		\end{array}
 		\right)}\vspace{0.2cm}
 	\end{equation*}
 	With
 	\begin{equation*}
 	A_R(X_{15})={\left(
 		\begin{array}{cc}
 		-0.0081 & 0.0106 \\
 		0.0265 & 0.0145\\
 		\end{array}
 		\right)}
 	\end{equation*}
 \end{example}

 \subsubsection{Reformulation of the two stage Block Horner method}
 After back substitution of the nested programmed scheme and accumulation we obtain:
 \begin{table}[H]
 	{\begin{tabular}{c}
 			$B_l(k)=A_R (X_k )$ \hspace{0.1cm} \mbox{and} \\
 			$C_{l-1}(k)=l{X_k}^{l-1}+{l-1}A_1 {X_k}^{l-2}+...+A_{l-1}=\Delta (X_k )$ \vspace{-0.3cm}
 	\end{tabular}}
 \end{table}
 \hspace*{-0.4cm} The two stage Block Horner Varian algorithm can be obtained when we use the compact forms of the matrices
 $B_l (k)$  and $C_{l-1}(k)$ in term of   $A_R$     lead us to Newton like iterated process.\\  [.5cm]
 
 \hspace{-.5cm}\textbf{Algorithm:} \vspace{-.2cm}
 \begin{figure}[H]
 	\includegraphics[width=10cm,height=8.5cm]{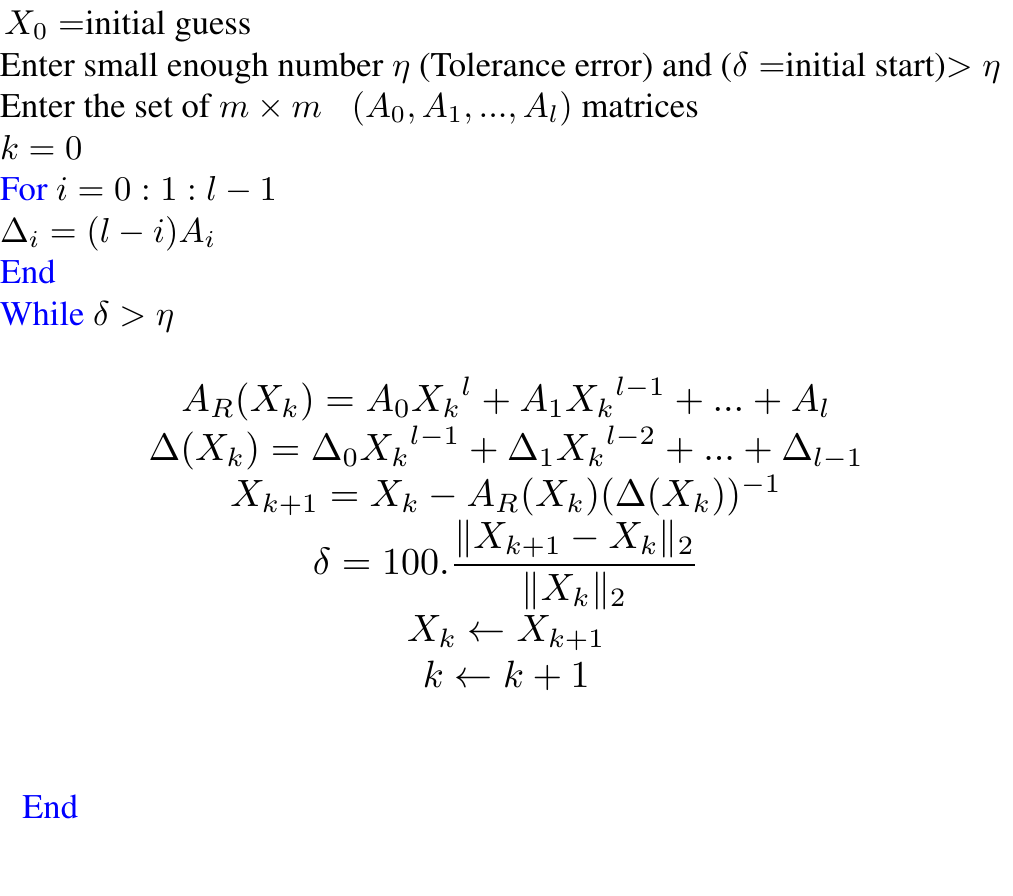}\vspace{-1cm}
 \end{figure}

 \subsection{Comments}

 \begin{itemize}
 	\item A numerical method for solving a given problem is said to be local if it is based on local (simpler)
 	model of the problem around the solution. From this definition, we can see that in order to use a local method, one has to provide an initial approximation of the solution. This initial approximation can be provided by a global method. As shown in Dahimene [\cite{09}], local methods are fast converging while global ones are quite slow. This implies that a good strategy is to start solving the problem by using a global method and then refine the solution by a local method.
 	
 	\item The proposed hybrid or two stage Block-Horner's algorithm converges rapidly as it performs a recursive iteration
 	and is easily implemented in a digital computer. Horner's algorithm could be used for evaluation of solvents of a matrix polynomial, but this method depends largely upon the initial guess even in some cases the initial value of $X_k$  is randomly chosen. Hence in sometimes it is very hard to find suitable solutions. Our  $Q.D.$ algorithm is numerically more stable and its initial starting values are well defined and evaluated.
 	
 	\item The complete program starts with the $Q.D.$ algorithm. It is then followed by a refinement of the right factor
 	by Horner's algorithm. After deflation, Horner’s algorithm is again applied using the next $Q$ output from the
 	$Q.D.$ algorithm and the process is repeated until we find a linear term. The above  process can
 	be applied only to polynomial matrices that satisfy the conditions of theorem (i.e. complete right and left
 	factorization and complete dominance relation between solvents). 
 	
 	

 	\item Many research works have been done on the spectral decomposition for matrix polynomials to achieve complete factorization and reconstruction of the block roots using algebraic and geometric numerical approaches, but (to our knowledge) nothing has been done for Block-Horner's algorithm and/or Block-$Q.D.$ algorithm.
 \end{itemize}
 
 \section{Application in control engineering}
 The system under examination is a power plant gas turbine (GE MS9001E) with single shaft, used as an electricity generator, installed in power station unit Sonelgaz at M'SILA, Algeria. The dynamic model of this gas turbine obtained via MIMO Recursive Least square estimator, using experimental inputs/outputs data acquired on-site and the obtained model is of order $n$=6 with two  inputs: (Output Pressure Compressor (OPC), and  Output Temperature Compressor (OTC)), and two outputs: (Exhaust Temperature and Rotor Speed) [\cite{44}]. 
 In figure 2, the fundamental components of the system under study are given .
 \begin{figure}[H]
 	\hspace{-0.4cm}
 	\includegraphics[scale=.57]{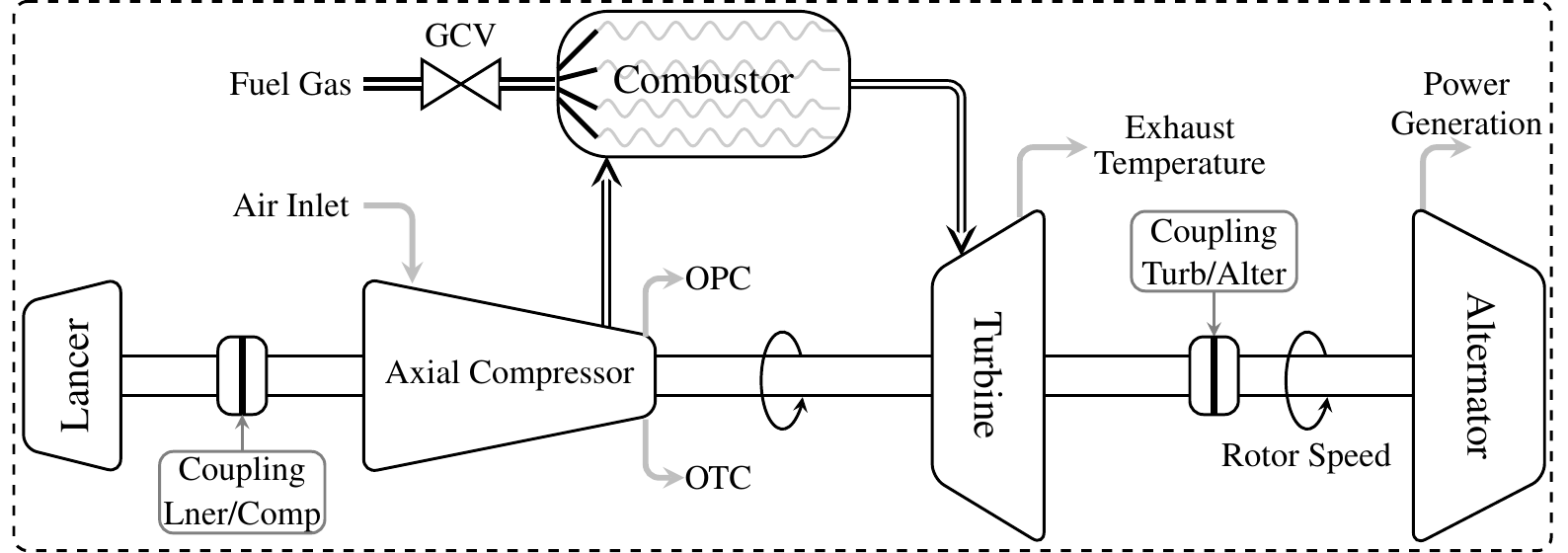}\\
 	\caption{schematic }
 \end{figure}
 The dynamic model of this power plant gas turbine is a linear time invariant multi input multi output system,described by a set of high degree coupled vector differential equations with matrix constant coefficients(or a matrix transfer function). In our case the relationship between the input and output is a ratio of two matrix polynomials, expressed as a right (or left) matrix fraction description (RMFD or LMFD):
 \begin{equation}
 \left\{
 \begin{array}{ll}
 H(\lambda)= N_{R}(\lambda){D_{R}}^{-1}(\lambda) \\
 \hspace{0.85cm}={D_{L}}^{-1}(\lambda)N_{L}(\lambda)
 \end{array}
 \right. 
 \end{equation}
 
 where:$N_{R},D_{R},N_{L} \mbox{and}  D_{L}$ are matrix polynomials and $\lambda$ stands for the $\left(\displaystyle{\frac{d}{dt}}\right)$ operator.  see  [\cite{01}],[\cite{02}] and [\cite{43}] and the reference therein .
 The obtained $\lambda$-matrix transfer function of the power plant gas turbine system is:
 \begin{equation*}
 H(\lambda)=N(\lambda){D(\lambda)}^{-1}=\left(
 \begin{array}{cc}
 H_{11}(\lambda) &H_{12}(\lambda) \\
 H_{21}(\lambda)& H_{22}(\lambda)\\
 \end{array}
 \right)
 \end{equation*}
 Where
 \begin{figure}[H]
 	\hspace{-0.3cm}
 	\includegraphics[width=9.1cm,height=4.7cm]{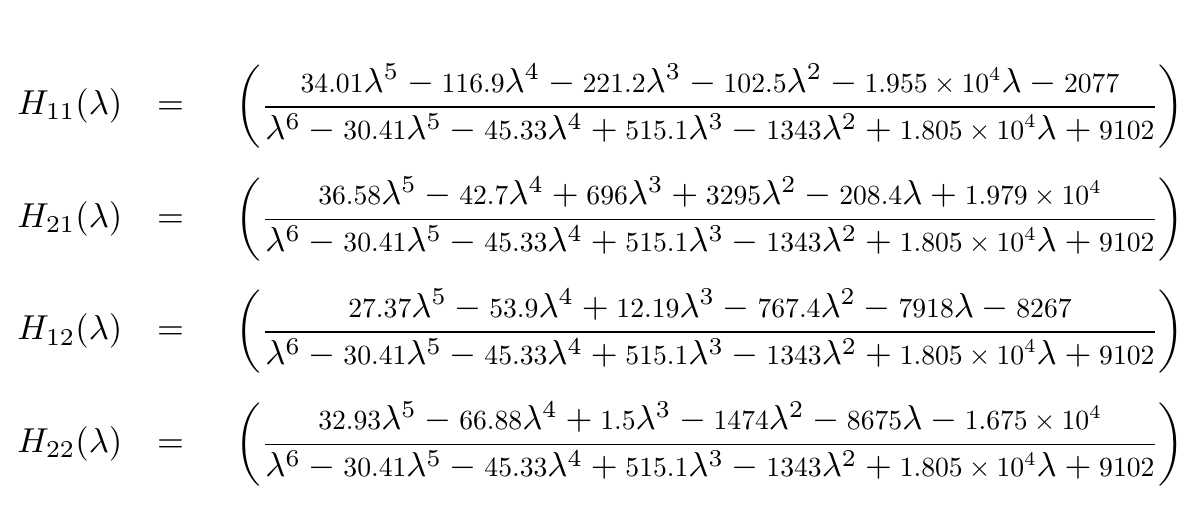}
 \end{figure}
 
 We try to decouple the power plant gas turbine dynamic model. Let us first factorize (\textbf{\emph{decompose}}) the numerator matrix polynomial $N(\lambda)$  into a complete set of spectral factors, then we use those block zeros into the denominator $D(\lambda)$ via state feedback. Hence the decoupling objectives are achieved.\\ [.15cm]
 Consider the square matrix transfer function:
 \begin{equation*}
 \begin{array}{ccc}
 H(\lambda)&=&N(\lambda)D^{-1} (\lambda)=
 \left(
 \begin{array}{cc}
 \sum\limits_{i=0}^{k} N_i \lambda^i \\
 \end{array}
 \right)
 \left(
 \begin{array}{cc}
 \sum\limits_{i=0}^{l} D_i \lambda^i
 \end{array}
 \right) ^{-1}\\
 &=&(N_k \lambda^k+...+N_1 \lambda+N_0 ) (D_l \lambda^l+...+D_1 \lambda+D_0 )^{-1}
 \end{array}
 \end{equation*}
 with: \\ \vspace{0.2cm}
 $D_l=I$ is an $ m\times m$ identity matrix and \\
 $N_i\in R^{m\times m},(i=0,1,...,k)$ \\
 $D_i\in R^{m\times m},(i=0,1,...,l)$, $l>k$ \\ [.25cm]
 Assume that $N(\lambda)$  can be factorized into $k$ Block zeros and $D(\lambda)$can be factorized into $l$ Block roots (\emph{using one of the proposed algorithms}): \\
 \begin{equation}
 N(\lambda)=N_k (\lambda I-Z_1 )...(\lambda I-Z_k ) \vspace{-0.15cm}
 \end{equation}
 \begin{equation}
 D(\lambda)=(\lambda I-Q_1 )...(\lambda I-Q_l ).\hspace{0.15cm}
 \end{equation}
 
 The matrix transfer function can be written : $H(\lambda)=C(\lambda I-A)^{-1}B$ . Also  via the use of state feedback the control law becomes a state dependent and be rewritten as  $u(t)=-K.X(t)+F.r(t)$. Hence we obtain the following closed loop system:
 \begin{equation*}
 (H(\lambda))_{closed}=C(\lambda I-A+BK)^{-1} BF=N(\lambda) D_d^{-1}(\lambda)F
 \end{equation*}
 where:  $D_d (\lambda)=(\lambda I-Q_{d1} )...(\lambda I-Q_{dl} )$  and $Q_{di}$: are the desired spectral factors to be placed
 \begin{equation*}
 H(\lambda)_{closed}=N(\lambda) D_d^{-1}(\lambda)F
 \end{equation*}
 Hence, the closed loop matrix transfer function is of the form:
 \begin{equation}
 H(\lambda)_{closed}=N_k (\lambda I-Z_1 )...(\lambda I-Z_k )(\lambda I-Q_{dl} )^{-1}...(\lambda I-Q_{d1} )^{-1}F
 \end{equation}
 In order to achieve perfect block decoupling we choose:\\\\
 $Q_{d1}=N_k J_1N_k^{-1},...,Q_{d(l-k)} =N_k J_{(l-k)}N_k^{-1}$\\
 $Q_{d(l-k+1)}=Z_1,...,Q_{dl}=Z_k$ \\
 $J_i=diag(\lambda_{i1},...,\lambda_{im}),\hspace{0.5cm} F=(N_k)^{-1}$\\ \\
 Now by assigning those block roots the system is decoupled and the closed loop matrix transfer function is:
 \begin{equation}
 H(\lambda)_{closed}=(\lambda I-J_{1} )^{-1}...(\lambda I-J_{l-k} )^{-1}
 \end{equation}
 Now we should construct the numerator and denominator matrix polynomials of the gas turbine system from the matrix transfer function (see [\cite{13},\cite{14} and \cite{27}]):
 \begin{equation*}
 \left\{
 \begin{array}{ccc}
 D(\lambda)&=&D_3\lambda^3+D_2\lambda^2+D_1\lambda+D_0, \hspace{1cm} D_3=I_2\\
 N(\lambda)&=&N_3\lambda^3+N_2\lambda^2+N_1\lambda+N_0, \hspace{1cm} N_3=O_2
 \end{array}
 \right.
 \end{equation*}
 where:
 
 \begin{equation*}
 \left( \begin{matrix}
 D_0 \\ \\ D_1\\ \\ D_2
 \end{matrix} \right )=
 \left (
 \begin{matrix}
 -37.0170&28.2888\\
 -223.8750&-74.7887\\
 34.8029&20.9798\\
 -280.2609&-216.8345\\
 -14.0378&-7.5183\\
 -12.3898&-16.3740
 \end{matrix}	
 \right)
 \end{equation*}
 
 \begin{equation*}
 \begin{array}{ccc}
 \left(\begin{array}{c}
 N_0\\ [.2cm] N_1\\ [.2cm] N_2\end{array}\right) =
 {
 	\left(
 	\begin{array}{cc}
 	211.7886&61.4727\\
 	331.6250&199.1758\\
 	100.8960&74.4572 \\
 	148.1818&120.4265\\
 	34.0105&27.3669\\
 	36.5764&32.9324\\
 	\end{array}
 	\right)}
 \end{array}
 \end{equation*}
 Let we decompose the numerator and denominator matrix polynomials and reconstruct their block roots:\\ [.2cm]
 $N(\lambda)=N_2(\lambda I-Z_1)(\lambda I-Z_2) $\\ [0.2cm]
 and \\ [0.2cm]
 $D(\lambda)=(\lambda I-Q_1)(\lambda I-Q_2)(\lambda I-Q_3)$ \\ [0.2cm]
 The Block spectral factors are approximately computerized with a residual normed tolerance error given by:\\ [0.3cm]
 $\hspace*{2cm} \displaystyle{\varepsilon_i=\frac{\|{Z_i}^*\|-\|Z_i\|}{\|{Z_i}^*\|}} \hspace{.15cm} i=1,2.$ \\ [0.2cm]
 and \\ [0.2cm]
 $\hspace*{2cm} \displaystyle{\xi_i=\frac{\|{Q_i}^*\|-\|Q_i\|}{\|{Q_i}^*\|}}  \hspace{.2cm} i=1,2,3$\\ [0.35cm]
 
 \begin{remark} The last Block pole $Q_3$ can be constructed using the synthetic long division.
 	Figure (3) illustrates a comparison  between the proposed algorithms in term of the convergence speed and residual normed tolrance error.\\ \end{remark}
 
 The numerator block zeros are computed using the proposed algorithms compared to recent developed method called the generalized secant method [\cite{45}] which can factorize matrix polynomial into a complete set of block roots. Numerical results of the developed procedures as illustrated in [\cite{45}] give:
 \begin{equation*}
 N(Z_i)=O_2\Rightarrow
 \end{equation*}
 \begin{equation*}
 Z_1 = {\left(
 	\begin{array}{cc}
 	24.7235 & 23.1394 \\
 	-27.4494 & -24.9281 \\
 	\end{array}
 	\right)}, \hspace{.5cm}
 Z_2 = {\left(
 	\begin{array}{cc}
 	-18.5711 & -16.0841 \\
 	16.1166 & 13.4353\\
 	\end{array}
 	\right)}
 \end{equation*}
 \begin{figure}[H]
 	\centering
 	\includegraphics[width=8.7cm,height=7.5cm]{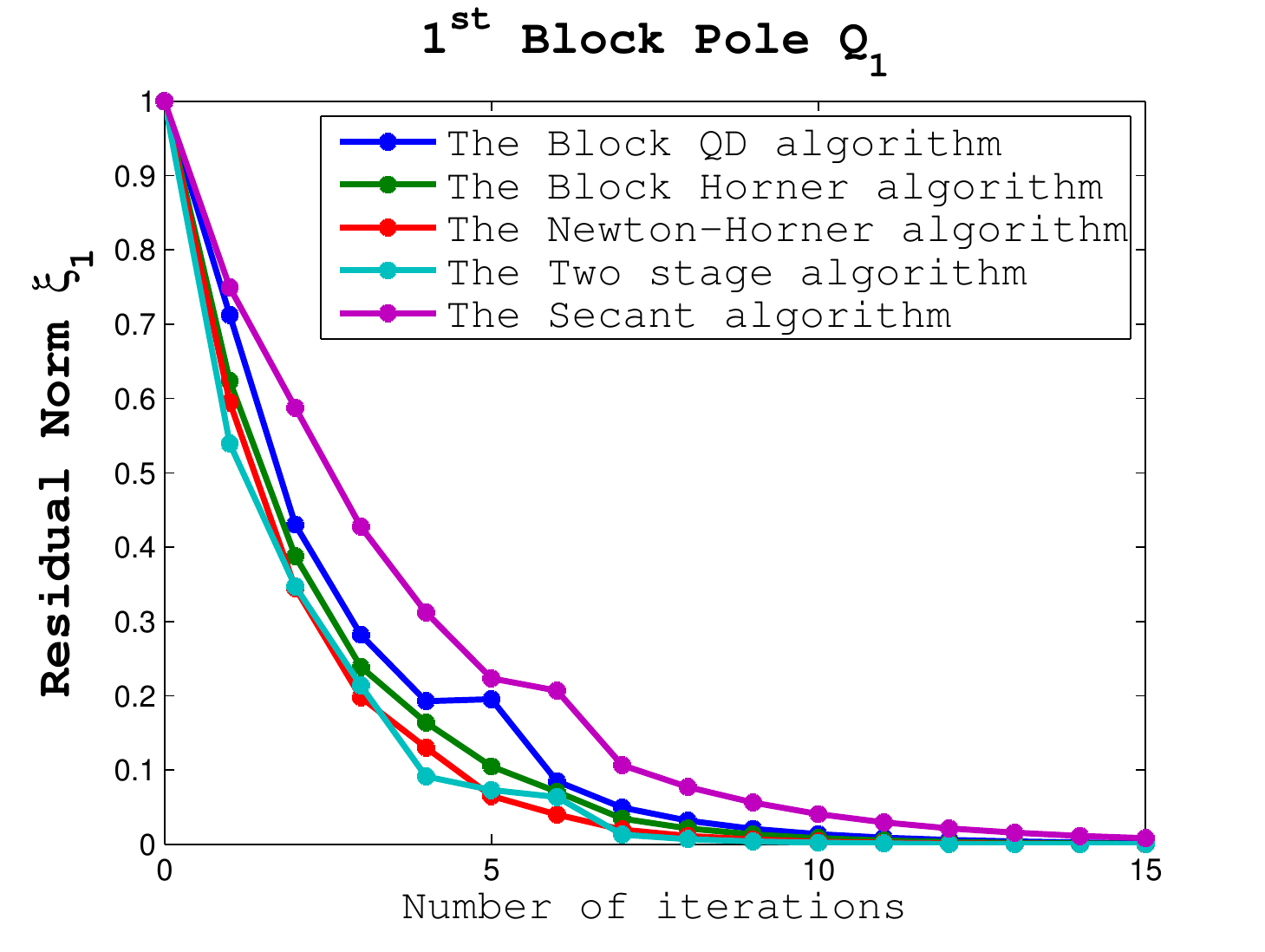}
 \end{figure}
 \begin{figure}[H]
 	\centering
 	\includegraphics[width=8.5cm,height=7.2cm]{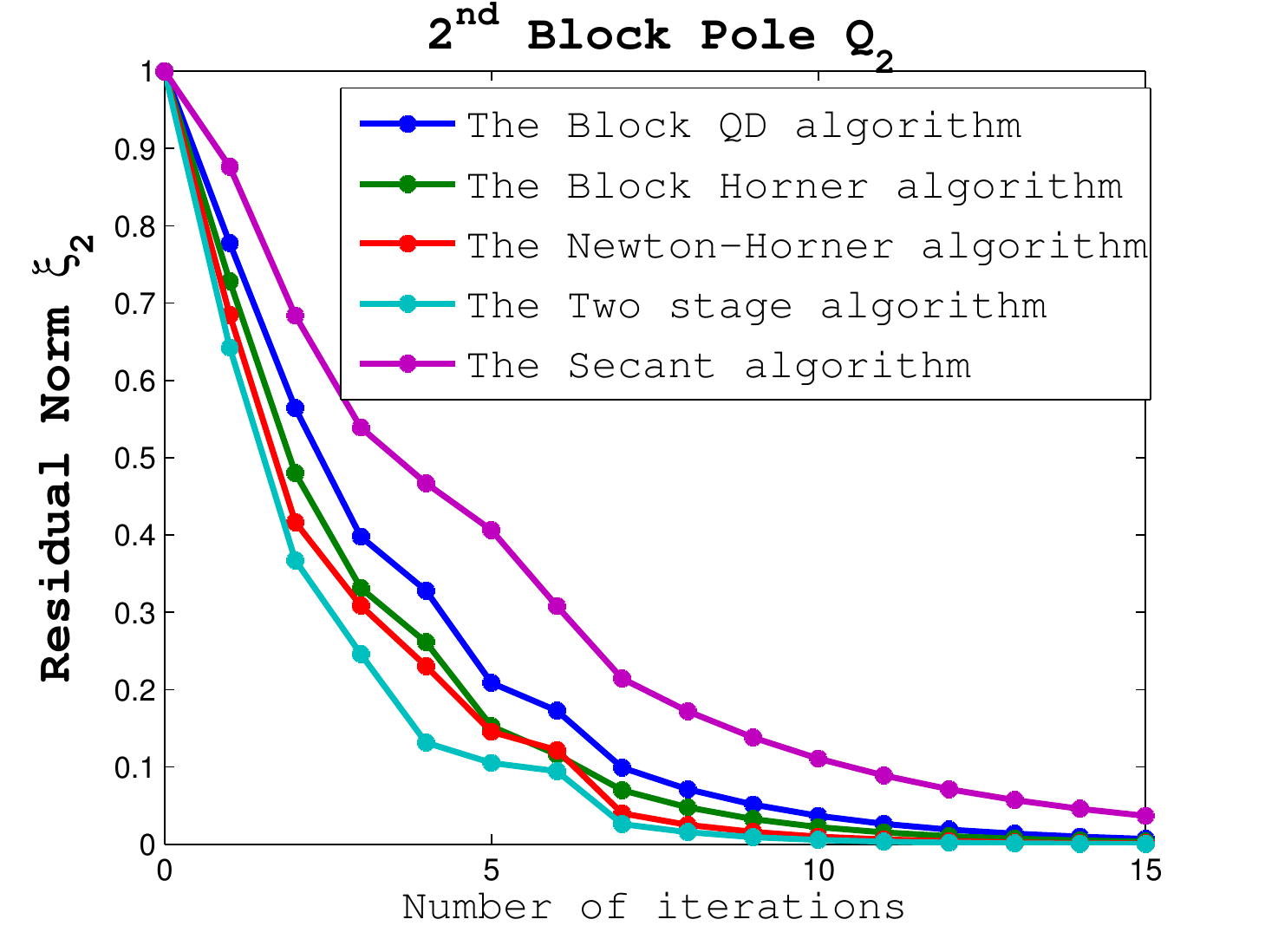}
 \end{figure}
 \begin{figure}[H]
 	\centering
 	\includegraphics[width=8.5cm,height=7.2cm]{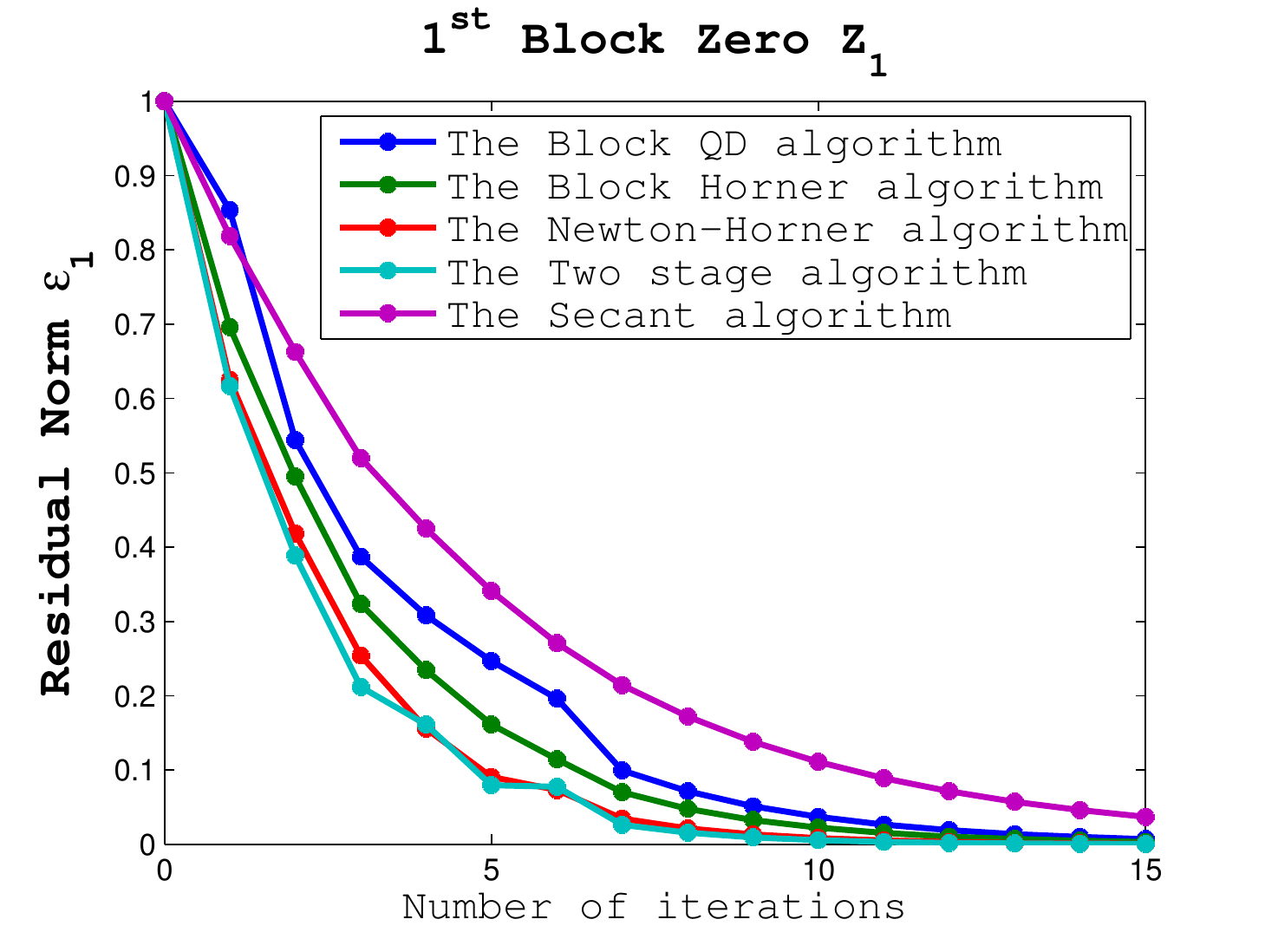}
 \end{figure}
 \begin{figure}[H]
 	\centering
 	\includegraphics[width=8.5cm,height=7.2cm]{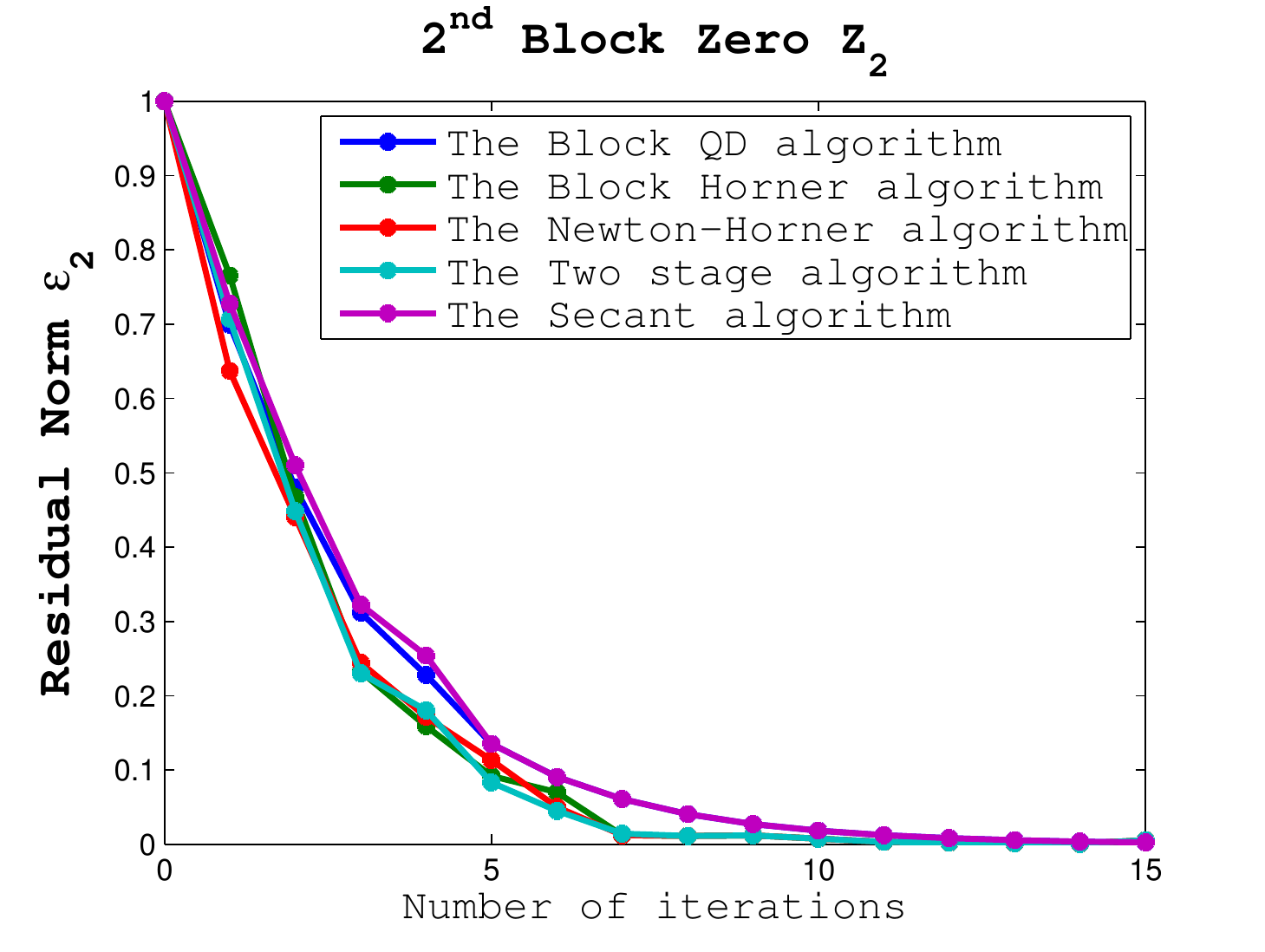}\\
 	\caption{The Residual Error Norm comparison study}
 \end{figure}
 
 Although the Block Newton method aims to improve the convergence speed over the Block Horner method, it cannot always achieve this goal. The Newton-Horner's method converges quadratically due to its conformity to Newton method. As a consequence, the number of significant values is roughly doubled every iteration, provided that $X_i$ is close to the root $(\varepsilon=0.521\times10^{-4})$.  The two stage algorithm is the best method of finding roots, it is simple and fast (at first seven iterations the average error becomes $\varepsilon=0.213\times10^{-3}$ and $\xi=0.341\times10^{-3}$). The only drawback of the two stage method is that it uses the matrix inversion, and partially is dependent on the initial guess. Our refinement algorithm  avoided this obstacle. As indicated in figure (2) the Q-D algorithm converges, but it is of global nature with no initial independence.  The global convergence characteristics of the secant method are poor, as indicated in this figure.\\ [0.2cm]
 The desired denominator is of third order written in the form:
 \begin{equation*}
 D_d (\lambda)=D_{d3} \lambda^3+D_{d2} \lambda^2+D_{d1} \lambda+D_{d0}
 \end{equation*}
 Using the prescribed decoupling algorithm we obtain:\\ [.2cm]
 $F={N_2}^{-1},\hspace{.05cm} J_1= {\left(\begin{array}{cc}-1 & 0\\0 & -2 \\ \end{array}\right)}$\\
 $Q_{d1}={N_2}^{-1}J_1{N_2}^{-1},\hspace{.05cm}  Q_{d2}=Z_1, \hspace{.05cm}  Q_{d3}=Z_2$
 \begin{equation*}
 \begin{array}{ccc}
 D_d (\lambda)&=&(\lambda I-Q_{d1})(\lambda I-Q_{d2})(\lambda I-Q_{d3})\\
 &=&I\lambda^3+D_{d2} \lambda^2+D_{d1}\lambda +D_{d0}
 \end{array}
 \end{equation*}
 Where:\\ [0.2cm]
 \begin{equation*}
 \begin{array}{ccc}
 D_{d2} &=& -(Q_{d1}+Q_{d2}+Q_{d3}) \\
 ~~ &=& {\left(\begin{array}{cc}-13.5596 & -14.6249\\21.7809 & 21.8999 \\\end{array}\right)} 
 \end{array}
 \end{equation*}
 \begin{equation*}
 \begin{array}{ccc}
 D_{d1}&=&(Q_{d1}Q_{d2}+Q_{d1}Q_{d3}+Q_{d2}Q_{d3})\\
 &=& {\left(\begin{array}{cc}-126.4282 & -121.5061\\161.6710 & 152.4741 \\\end{array}\right)}\\
 \end{array}
 \end{equation*}
 \begin{equation*}
 \begin{array}{ccc}
 D_{d0}&=&-(Q_{d1}Q_{d2}Q_{d3})\\
 &=& {\left(\begin{array}{cc}-178.9732 & -164.0512\\223.2851 & 202.6227 \\\end{array}\right)}
 \end{array}
 \end{equation*}
 The state feedback gain matrix of the Block controller form is obtained by see [\cite{42}],[\cite{43}]:
 \begin{equation*}
 K_{ci}=D_{di}-D_{i} \hspace{0.15cm} \mbox{With}\hspace{0.15cm} i=0,1,2  \hspace{0.15cm} \hspace{0.15cm} \mbox{and}\hspace{0.15cm} K_{c}=[ K_{c0}, K_{c1}, K_{c2}]
 \end{equation*}
 Now let we go back to original base by similarity transformation $T_c$ as found in [\cite{01}],[\cite{02}] and [\cite{43}]:
 \begin{equation*}
 {
 	\begin{array}{ccc}
 	K&=&K_{c} T_c\\
 	&=&\left(
 	\begin{array}{cccccc}
 	-0.7616&-3.2690&3.5737&-0.0716&-2.3462&1.4801\\
 	3.0439&5.4047&-2.3853&2.4117&4.2560&0.0494\\
 	\end{array}
 	\right)
 	\end{array}}
 \end{equation*}
 The new model of the decoupled system after state feedback is:
 \begin{equation*}
 A_d=(A-BK), \hspace{0.5cm}  B_d=BF, \hspace{0.5cm}  \mbox{and} \hspace{0.5cm} C_d=C
 \end{equation*}
 \begin{equation*}
 \begin{array}{ccc}
 H(\lambda)_{closed}&=&C(\lambda I-A+BK)^{-1} BF \hspace*{2cm} \\
 &=&N(\lambda) D_d^{-1}(\lambda)F=
 {\left(
 	\begin{array}{cc}
 	\displaystyle{\frac{1}{\lambda+1}}& 0 \\
 	0 & \displaystyle{\frac{1}{\lambda+2}}\\
 	\end{array}
 	\right)}
 \end{array}
 \end{equation*}
 
 Based on the results we deduce that the Block roots are well computed, both numerator and denominator matrix polynomials ($N(\lambda)$ and $D(\lambda)$) are perfectly decomposed using the proposed procedure.
 
 \subsection{Suggestions for further research}
 The results obtained during this research work arose many questions and problems which are subject for future research
 \begin{itemize}
 	\item Finding other globalization techniques for the Block-Horner's algorithm to avoid the local restriction and the problem of initial guess, so to arrive at very fast global nested program. Also exploring and 	extending other scalar numerical methods to factorize matrix polynomials. 	\item Both of The Block-Horner's algorithm and the Block-$Q.D.$ algorithm as used in our work converges to
 	factors of a matrix polynomial. By using the defined similarity transformations, we can derive the solvents. However, it would be convenient to have a global algorithm that converges rapidly and directly to all solvents.
 	\item If a column in the $Q.D.$ tableau converges, it implies that there exists a factorization of the matrix polynomial that splits the spectrum into a dominant set and a dominated one. If the system under consideration is a digital system, we know that the largest modulus latent roots affect the dynamic properties of the system. In such case, the $Q.D.$ algorithm can become a tool for system 	reduction (using the dominant mode concept).
 	\item The computational procedure for finding the solvents of a matrix polynomial with repeated block roots (solvents) and/or spectral factors need to be investigated further.
 \end{itemize}
 
 \section{Conclusion}
 In this paper we have introduced new numerical approaches for determining the complete sets of spectral factors and
 solvents of a monic matrix polynomial. For avoiding the initial guess we have proposed a systematic method for
 the Block-Horner's algorithm via a refinement of the Block-$Q.D.$ algorithm. At least three advantages are offeedr by
 the proposed technique: (i) an algorithm with global nature is obtained; hence there is no initial-guess problem
 during the whole procedure, (ii) high speed convergence to each solution and only a few iterations are required
 (iii) via the help of refinement and direct cascading, the algorithms are easily coupled together and the whole
 scheme is suitable for programming in a digital computer digital. The obtained solvents can be
 considered as a useful tool for carrying out the block partial fraction expansion for the inverse of a matrix
 polynomial. Those partial fractions are matrix transfer functions of reduced order linear systems such
 that the realization of them leads to block diagonal (block-decoupling) or parallel decomposed multivariable linear
 time invariant system. The dynamic properties of MIMO systems depend on block pole of its characteristic matrix polynomial. Therefore they can be used as tools for block-pole placement, block-system identification and block-model order reduction. In addition, the proposed
 method can be employed to carry out the block spectral factorization of a matrix polynomial for problems in optimal
 control, filtering and estimation.\\ 
 
 
 \section*{Acknowledgement}
 Dr. George F. Fragulis work is supported by the program Rescom 80159 of the Univ. of Applied Science of Western Macedonia, Hellas.


\end{document}